\documentclass[11pt, reqno]{amsart}
\usepackage{txfonts}    
\usepackage{amssymb}
\usepackage{eucal}
\usepackage{amsmath}
\usepackage{amsthm}
\usepackage{amscd}
\usepackage[dvips]{color}
\usepackage{multicol}
\usepackage[all]{xy}        
\usepackage{graphicx}
\usepackage{color}
\usepackage{colordvi}
\usepackage{xspace}
\usepackage{tikz}

\usepackage{ifpdf}
\ifpdf
 \usepackage[colorlinks,final,hyperindex]{hyperref}
\else
 \usepackage[colorlinks,final,hyperindex,hypertex]{hyperref}
\fi

\newcommand {\emptycomment}[1]{} 

\newcommand{\nc}{\newcommand}
\newcommand{\delete}[1]{}
\nc{\mfootnote}[1]{\footnote{#1}} 
\nc{\todo}[1]{\tred{To do:} #1}

\nc{\mlabel}[1]{\label{#1}}  
\nc{\mcite}[1]{\cite{#1}}  
\nc{\mref}[1]{\ref{#1}}  
\nc{\meqref}[1]{\eqref{#1}} 
\nc{\mbibitem}[1]{\bibitem{#1}} 

\delete{
\nc{\mlabel}[1]{\label{#1}  
{\hfill \hspace{1cm}{\bf{{\ }\hfill(#1)}}}}
\nc{\mcite}[1]{\cite{#1}{{\bf{{\ }(#1)}}}}  
\nc{\mref}[1]{\ref{#1}{{\bf{{\ }(#1)}}}}  
\nc{\meqref}[1]{\eqref{#1}{{\bf{{\ }(#1)}}}} 
\nc{\mbibitem}[1]{\bibitem[\bf #1]{#1}} 
}
\makeatletter
\@namedef{subjclassname@2020}{\textup{2020} Mathematics Subject Classification}
\makeatother

\newtheorem{thm}{Theorem}[section]
\newtheorem{lem}[thm]{Lemma}

\newtheorem{cor}[thm]{Corollary}
\newtheorem{pro}[thm]{Proposition}

\theoremstyle{definition}
\newtheorem{defi}[thm]{Definition}
\newtheorem{ex}[thm]{Example}
\newtheorem{rmk}[thm]{Remark}

\newtheorem{defil}[thm]{Definition-Lemma}

\nc{\tred}[1]{\textcolor{red}{#1}}
\nc{\tblue}[1]{\textcolor{blue}{#1}}
\nc{\tgreen}[1]{\textcolor{green}{#1}}
\nc{\tpurple}[1]{\textcolor{purple}{#1}}
\nc{\btred}[1]{\textcolor{red}{\bf #1}}
\nc{\btblue}[1]{\textcolor{blue}{\bf #1}}
\nc{\btgreen}[1]{\textcolor{green}{\bf #1}}
\nc{\btpurple}[1]{\textcolor{purple}{\bf #1}}

\nc{\si}[1]{\textcolor{purple}{Siyuan:#1}}
\nc{\cm}[1]{\textcolor{red}{Chengming:#1}}
\nc{\hs}[1]{\textcolor{red}{#1}}
\nc{\lir}[1]{\textcolor{blue}{Li:#1}}

\nc{\twovec}[2]{\left(\begin{array}{c} #1 \\ #2\end{array} \right )}
\nc{\threevec}[3]{\left(\begin{array}{c} #1 \\ #2 \\ #3 \end{array}\right )}
\nc{\twomatrix}[4]{\left(\begin{array}{cc} #1 & #2\\ #3 & #4 \end{array} \right)}
\nc{\threematrix}[9]{{\left(\begin{matrix} #1 & #2 & #3\\ #4 & #5 & #6 \\ #7 & #8 & #9 \end{matrix} \right)}}
\nc{\twodet}[4]{\left|\begin{array}{cc} #1 & #2\\ #3 & #4 \end{array} \right|}

\nc{\rk}{\mathrm{r}}

\newcommand{\A}{A_h}

\nc{\tforall}{\text{ for all }}

\nc{\svec}[2]{{\tiny\left(\begin{matrix}#1\\
#2\end{matrix}\right)\,}}  
\nc{\ssvec}[2]{{\tiny\left(\begin{matrix}#1\\
#2\end{matrix}\right)\,}} 

\nc{\typeI}{local cocycle $3$-Lie bialgebra\xspace}
\nc{\typeIs}{local cocycle $3$-Lie bialgebras\xspace}
\nc{\typeII}{double construction $3$-Lie bialgebra\xspace}
\nc{\typeIIs}{double construction $3$-Lie bialgebras\xspace}

\nc{\bia}{{$\mathcal{P}$-bimodule ${\bf k}$-algebra}\xspace}
\nc{\bias}{{$\mathcal{P}$-bimodule ${\bf k}$-algebras}\xspace}
\nc{\rmi}{{\mathrm{I}}}
\nc{\rmii}{{\mathrm{II}}}
\nc{\rmiii}{{\mathrm{III}}}
\nc{\pr}{{\mathrm{pr}}}

\nc{\OT}{constant $\theta$-}
\nc{\T}{$\theta$-}
\nc{\IT}{inverse $\theta$-}

\nc{\asi}{ASI\xspace}
\nc{\qadm}{$Q$-admissible\xspace}
\nc{\aybe}{AYBE\xspace}
\nc{\admset}{\{\pm x\}\cup (-x+K^{\times}) \cup K^{\times} x^{-1}}

\nc{\dualrep}{gives a dual representation\xspace}
\nc{\admt}{admissible to\xspace}

\nc{\opa}{\cdot_A}
\nc{\opb}{\cdot_B}

\nc{\post}{positive type\xspace}
\nc{\negt}{negative type\xspace}
\nc{\invt}{inverse type\xspace}

\nc{\pll}{\beta}
\nc{\plc}{\epsilon}

\nc{\ass}{{\mathit{Ass}}}
\nc{\lie}{{\mathit{Lie}}}
\nc{\comm}{{\mathit{Comm}}}
\nc{\dend}{{\mathit{Dend}}}
\nc{\zinb}{{\mathit{Zinb}}}
\nc{\tdend}{{\mathit{TDend}}}
\nc{\prelie}{{\mathit{preLie}}}
\nc{\postlie}{{\mathit{PostLie}}}
\nc{\quado}{{\mathit{Quad}}}
\nc{\octo}{{\mathit{Octo}}}
\nc{\ldend}{{\mathit{ldend}}}
\nc{\lquad}{{\mathit{LQuad}}}

 \nc{\adec}{\check{;}} \nc{\aop}{\alpha}
\nc{\dftimes}{\widetilde{\otimes}} \nc{\dfl}{\succ} \nc{\dfr}{\prec}
\nc{\dfc}{\circ} \nc{\dfb}{\bullet} \nc{\dft}{\star}
\nc{\dfcf}{{\mathbf k}} \nc{\apr}{\ast} \nc{\spr}{\cdot}
\nc{\twopr}{\circ} \nc{\tspr}{\star} \nc{\sempr}{\ast}
\nc{\disp}[1]{\displaystyle{#1}}
\nc{\bin}[2]{ (_{\stackrel{\scs{#1}}{\scs{#2}}})}  
\nc{\binc}[2]{ \left (\!\! \begin{array}{c} \scs{#1}\\
    \scs{#2} \end{array}\!\! \right )}  
\nc{\bincc}[2]{  \left ( {\scs{#1} \atop
    \vspace{-.5cm}\scs{#2}} \right )}  
\nc{\sarray}[2]{\begin{array}{c}#1 \vspace{.1cm}\\ \hline
    \vspace{-.35cm} \\ #2 \end{array}}
\nc{\bs}{\bar{S}} \nc{\dcup}{\stackrel{\bullet}{\cup}}
\nc{\dbigcup}{\stackrel{\bullet}{\bigcup}} \nc{\etree}{\big |}
\nc{\la}{\longrightarrow} \nc{\fe}{\'{e}} \nc{\rar}{\rightarrow}
\nc{\dar}{\downarrow} \nc{\dap}[1]{\downarrow
\rlap{$\scriptstyle{#1}$}} \nc{\uap}[1]{\uparrow
\rlap{$\scriptstyle{#1}$}} \nc{\defeq}{\stackrel{\rm def}{=}}
\nc{\dis}[1]{\displaystyle{#1}} \nc{\dotcup}{\,
\displaystyle{\bigcup^\bullet}\ } \nc{\sdotcup}{\tiny{
\displaystyle{\bigcup^\bullet}\ }} \nc{\hcm}{\ \hat{,}\ }
\nc{\hcirc}{\hat{\circ}} \nc{\hts}{\hat{\shpr}}
\nc{\lts}{\stackrel{\leftarrow}{\shpr}}
\nc{\rts}{\stackrel{\rightarrow}{\shpr}} \nc{\lleft}{[}
\nc{\lright}{]} \nc{\uni}[1]{\tilde{#1}} \nc{\wor}[1]{\check{#1}}
\nc{\free}[1]{\bar{#1}} \nc{\den}[1]{\check{#1}} \nc{\lrpa}{\wr}
\nc{\curlyl}{\left \{ \begin{array}{c} {} \\ {} \end{array}
    \right .  \!\!\!\!\!\!\!}
\nc{\curlyr}{ \!\!\!\!\!\!\!
    \left . \begin{array}{c} {} \\ {} \end{array}
    \right \} }
\nc{\leaf}{\ell}       
\nc{\longmid}{\left | \begin{array}{c} {} \\ {} \end{array}
    \right . \!\!\!\!\!\!\!}
\nc{\ot}{\otimes} \nc{\sot}{{\scriptstyle{\ot}}}
\nc{\otm}{\overline{\ot}}
\nc{\ora}[1]{\stackrel{#1}{\rar}}
\nc{\ola}[1]{\stackrel{#1}{\la}}
\nc{\pltree}{\calt^\pl}
\nc{\epltree}{\calt^{\pl,\NC}}
\nc{\rbpltree}{\calt^r}
\nc{\scs}[1]{\scriptstyle{#1}} \nc{\mrm}[1]{{\rm #1}}
\nc{\dirlim}{\displaystyle{\lim_{\longrightarrow}}\,}
\nc{\invlim}{\displaystyle{\lim_{\longleftarrow}}\,}
\nc{\mvp}{\vspace{0.5cm}} \nc{\svp}{\vspace{2cm}}
\nc{\vp}{\vspace{8cm}} \nc{\proofbegin}{\noindent{\bf Proof: }}
\nc{\proofend}{$\blacksquare$ \vspace{0.5cm}}
\nc{\freerbpl}{{F^{\mathrm RBPL}}}
\nc{\sha}{{\mbox{\cyr X}}}  
\nc{\ncsha}{{\mbox{\cyr X}^{\mathrm NC}}} \nc{\ncshao}{{\mbox{\cyr
X}^{\mathrm NC,\,0}}}
\nc{\shpr}{\diamond}    
\nc{\shprm}{\overline{\diamond}}    
\nc{\shpro}{\diamond^0} 
\nc{\shprr}{\diamond^r}  
\nc{\shpra}{\overline{\diamond}^r}
\nc{\shpru}{\check{\diamond}} \nc{\catpr}{\diamond_l}
\nc{\rcatpr}{\diamond_r} \nc{\lapr}{\diamond_a}
\nc{\sqcupm}{\ot}
\nc{\lepr}{\diamond_e} \nc{\vep}{\varepsilon} \nc{\labs}{\mid\!}
\nc{\rabs}{\!\mid} \nc{\hsha}{\widehat{\sha}}
\nc{\lsha}{\stackrel{\leftarrow}{\sha}}
\nc{\rsha}{\stackrel{\rightarrow}{\sha}} \nc{\lc}{\lfloor}
\nc{\rc}{\rfloor}
\nc{\tpr}{\sqcup}
\nc{\nctpr}{\vee}
\nc{\plpr}{\star}
\nc{\rbplpr}{\bar{\plpr}}
\nc{\sqmon}[1]{\langle #1\rangle}
\nc{\forest}{\calf}
\nc{\altx}{\Lambda_X} \nc{\vecT}{\vec{T}} \nc{\onetree}{\bullet}
\nc{\Ao}{\check{A}}
\nc{\seta}{\underline{\Ao}}
\nc{\deltaa}{\overline{\delta}}
\nc{\trho}{\tilde{\rho}}

\nc{\rpr}{\circ}
\nc{\dpr}{{\tiny\diamond}}
\nc{\rprpm}{{\rpr}}

\nc{\mmbox}[1]{\mbox{\ #1\ }} \nc{\ann}{\mrm{ann}}
\nc{\Aut}{\mrm{Aut}} \nc{\can}{\mrm{can}}
\nc{\twoalg}{{two-sided algebra}\xspace}
\nc{\colim}{\mrm{colim}}
\nc{\Cont}{\mrm{Cont}} \nc{\rchar}{\mrm{char}}
\nc{\cok}{\mrm{coker}} \nc{\dtf}{{R-{\rm tf}}} \nc{\dtor}{{R-{\rm
tor}}}

\nc{\depth}{{\mrm d}}
\nc{\Div}{{\mrm Div}} \nc{\End}{\mrm{End}} \nc{\Ext}{\mrm{Ext}}
\nc{\Fil}{\mrm{Fil}} \nc{\Frob}{\mrm{Frob}} \nc{\Gal}{\mrm{Gal}}
\nc{\GL}{\mrm{GL}} \nc{\Hom}{\mrm{Hom}} \nc{\hsr}{\mrm{H}}
\nc{\hpol}{\mrm{HP}} \nc{\id}{\mrm{id}} \nc{\im}{\mrm{im}}
\nc{\incl}{\mrm{incl}} \nc{\length}{\mrm{length}}
\nc{\LR}{\mrm{LR}} \nc{\mchar}{\rm char} \nc{\NC}{\mrm{NC}}
\nc{\mpart}{\mrm{part}} \nc{\pl}{\mrm{PL}}
\nc{\ql}{{\QQ_\ell}} \nc{\qp}{{\QQ_p}}
\nc{\rank}{\mrm{rank}} \nc{\rba}{\rm{RBA }} \nc{\rbas}{\rm{RBAs }}
\nc{\rbpl}{\mrm{RBPL}}
\nc{\rbw}{\rm{RBW }} \nc{\rbws}{\rm{RBWs }} \nc{\rcot}{\mrm{cot}}
\nc{\rest}{\rm{controlled}\xspace}
\nc{\rdef}{\mrm{def}} \nc{\rdiv}{{\rm div}} \nc{\rtf}{{\rm tf}}
\nc{\rtor}{{\rm tor}} \nc{\res}{\mrm{res}} \nc{\SL}{\mrm{SL}}
\nc{\Spec}{\mrm{Spec}} \nc{\tor}{\mrm{tor}} \nc{\Tr}{\mrm{Tr}}
\nc{\mtr}{\mrm{sk}}

\nc{\ab}{\mathbf{Ab}} \nc{\Alg}{\mathbf{Alg}}
\nc{\Algo}{\mathbf{Alg}^0} \nc{\Bax}{\mathbf{Bax}}
\nc{\Baxo}{\mathbf{Bax}^0} \nc{\RB}{\mathbf{RB}}
\nc{\RBo}{\mathbf{RB}^0} \nc{\BRB}{\mathbf{RB}}
\nc{\Dend}{\mathbf{DD}} \nc{\bfk}{{\bf k}} \nc{\bfone}{{\bf 1}}
\nc{\base}[1]{{a_{#1}}} \nc{\detail}{\marginpar{\bf More detail}
    \noindent{\bf Need more detail!}
    \svp}
\nc{\Diff}{\mathbf{Diff}} \nc{\gap}{\marginpar{\bf
Incomplete}\noindent{\bf Incomplete!!}
    \svp}
\nc{\FMod}{\mathbf{FMod}} \nc{\mset}{\mathbf{MSet}}
\nc{\rb}{\mathrm{RB}} \nc{\Int}{\mathbf{Int}}
\nc{\Mon}{\mathbf{Mon}}
\nc{\remarks}{\noindent{\bf Remarks: }}
\nc{\OS}{\mathbf{OS}} 
\nc{\Rep}{\mathbf{Rep}}
\nc{\Rings}{\mathbf{Rings}} \nc{\Sets}{\mathbf{Sets}}
\nc{\DT}{\mathbf{DT}}

\nc{\BA}{{\mathbb A}} \nc{\CC}{{\mathbb C}} \nc{\DD}{{\mathbb D}}
\nc{\EE}{{\mathbb E}} \nc{\FF}{{\mathbb F}} \nc{\GG}{{\mathbb G}}
\nc{\HH}{{\mathbb H}} \nc{\LL}{{\mathbb L}} \nc{\NN}{{\mathbb N}}
\nc{\QQ}{{\mathbb Q}} \nc{\RR}{{\mathbb R}} \nc{\BS}{{\mathbb{S}}} \nc{\TT}{{\mathbb T}}
\nc{\VV}{{\mathbb V}} \nc{\ZZ}{{\mathbb Z}}

\nc{\calao}{{\mathcal A}} \nc{\cala}{{\mathcal A}}
\nc{\calc}{{\mathcal C}} \nc{\cald}{{\mathcal D}}
\nc{\cale}{{\mathcal E}} \nc{\calf}{{\mathcal F}}
\nc{\calfr}{{{\mathcal F}^{\,r}}} \nc{\calfo}{{\mathcal F}^0}
\nc{\calfro}{{\mathcal F}^{\,r,0}} \nc{\oF}{\overline{F}}
\nc{\calg}{{\mathcal G}} \nc{\calh}{{\mathcal H}}
\nc{\cali}{{\mathcal I}} \nc{\calj}{{\mathcal J}}
\nc{\call}{{\mathcal L}} \nc{\calm}{{\mathcal M}}
\nc{\caln}{{\mathcal N}} \nc{\calo}{{\mathcal O}}
\nc{\calp}{{\mathcal P}} \nc{\calq}{{\mathcal Q}} \nc{\calr}{{\mathcal R}}
\nc{\calt}{{\mathcal T}} \nc{\caltr}{{\mathcal T}^{\,r}}
\nc{\calu}{{\mathcal U}} \nc{\calv}{{\mathcal V}}
\nc{\calw}{{\mathcal W}} \nc{\calx}{{\mathcal X}}
\nc{\CA}{\mathcal{A}}

\nc{\fraka}{{\mathfrak a}} \nc{\frakB}{{\mathfrak B}}
\nc{\frakb}{{\mathfrak b}} \nc{\frakd}{{\mathfrak d}}
\nc{\oD}{\overline{D}}
\nc{\frakF}{{\mathfrak F}} \nc{\frakg}{{\mathfrak g}}
\nc{\frakm}{{\mathfrak m}} \nc{\frakM}{{\mathfrak M}}
\nc{\frakMo}{{\mathfrak M}^0} \nc{\frakp}{{\mathfrak p}}
\nc{\frakS}{{\mathfrak S}} \nc{\frakSo}{{\mathfrak S}^0}
\nc{\fraks}{{\mathfrak s}} \nc{\os}{\overline{\fraks}}
\nc{\frakT}{{\mathfrak T}}
\nc{\oT}{\overline{T}}
\nc{\frakX}{{\mathfrak X}} \nc{\frakXo}{{\mathfrak X}^0}
\nc{\frakx}{{\mathbf x}}
\nc{\frakTx}{\frakT}      
\nc{\frakTa}{\frakT^a}    
\nc{\frakTxo}{\frakTx^0}   
\nc{\caltao}{\calt^{a,0}}   
\nc{\ox}{\overline{\frakx}} \nc{\fraky}{{\mathfrak y}}
\nc{\frakz}{{\mathfrak z}} \nc{\oX}{\overline{X}}

\font\cyr=wncyr10

\nc{\al}{\alpha}
\nc{\lam}{\lambda}
\nc{\lr}{\longrightarrow}

\usepackage{ulem} 
\allowdisplaybreaks 
\nc{\lyc}[1]{\textcolor{blue}{Lin Yuanchang: #1}}

\begin{document}

\title{Quantizations of Transposed Poisson algebras by Novikov deformations}



\author{Siyuan Chen}
\address{Chern Institute of Mathematics \& LPMC, Nankai University, Tianjin 300071, China}
\email{1120210010@mail.nankai.edu.cn}

\author{Chengming Bai$^{*}$}
\thanks{*Author to whom any correspondence should be addressed.}
\address{Chern Institute of Mathematics \& LPMC, Nankai University, Tianjin 300071, China}
\email{baicm@nankai.edu.cn}

\subjclass[2020]{13D10, 13N15, 17A30, 17B63, 53D55}

\keywords{transposed Poisson algebra; Novikov algebra; deformation; quantization.}


\begin{abstract}
 The notions of the Novikov deformation of a commutative associative algebra and the corresponding classical limit are introduced. We show such a classical limit belongs to a subclass of transposed Poisson
 algebras, and hence the Novikov deformation is defined to be the quantization of the corresponding
  transposed Poisson algebra. As a direct
  consequence, we revisit the relationship between transposed
  Poisson algebras and Novikov-Poisson algebras due to the fact that there is a natural Novikov
deformation of the commutative associative algebra in a
Novikov-Poisson algebra. Hence all transposed Poisson algebras of
Novikov-Poisson type, including unital transposed Poisson
algebras, can be quantized. Finally, we classify the
quantizations of $2$-dimensional complex transposed Poisson
algebras in which the Lie brackets are non-abelian up to
equivalence.

\end{abstract}

\maketitle

\tableofcontents

\section{Introduction}
The notion of a transposed Poisson algebra was introduced in
\cite{BBGW-TPP} by exchanging the roles of the two binary
operations in the Leibniz rule defining a Poisson algebra.
Transposed Poisson algebras are related to various algebraic
structures. In \cite{BBGW-TPP}, the relations of transposed
Poisson algebras with Novikov-Poisson algebras and $3$-Lie
algebras were given. In \cite{Sar-cgd}, the results that the
variety of commutative Gel'fand-Dorfman algebras coincides with
the variety of transposed Poisson algebras as well as transposed
Poisson algebras are $F$-manifold algebras were obtained. In
\cite{FKL}, a relation between $\frac{1}{2}$-derivations of Lie
algebras and transposed Poisson algebras was established, and
hence  along this approach, several interesting examples of
transposed Poisson algebras were constructed (\cite{K} and the
references therein).

It is well-known (\cite{MR-pol}) that Poisson algebras are
classical limits of associative deformations of commutative
associative algebras, and hence the associative deformations of
commutative associative algebras are defined to be the
quantizations of the corresponding Poisson algebras, that is,
Poisson algebras are quantized by associative deformations of
commutative associative algebras. It has been shown in
\cite{Do-fm,LSB-dpL} that $F$-manifold algebras are quantized by
pre-Lie deformations of commutative associative algebras. So it is
natural to consider whether transposed Poisson algebras can be
quantized by certain deformations of commutative associative
algebras. The present paper studies this question and gives a
partial answer by introducing the notion of Novikov deformations
of commutative associative algebras.

In fact, Novikov algebras were introduced in connection with the Hamiltonian operators in the formal variational calculus (\cite{GD}) and Poisson brackets of hydrodynamic type (\cite{BN-Nov}). The deformation theory of Novikov algebras has
been established in \cite{BM-DN} to study the ``realization"
theory of Novikov algebras. Note that a Novikov algebra is commutative if and only if it is a commutative associative algebra. Hence it makes sense to consider the deformations of commutative associative algebras as Novikov algebras, that is, Novikov deformations of commutative associative algebras. Similarly, it is also natural to introduce and then consider the
associated quantizations, which have never been studied before.

The main conclusion in this paper is that the classical limit of
the Novikov deformation of a commutative associative algebra
belongs to a subclass of transposed Poisson algebras. Therefore
the notion of a quantization of a transposed Poisson algebra is
introduced by Novikov deformations, that is, the Novikov
deformation of a commutative associative algebra is defined to be
the quantization of the corresponding transposed Poisson algebra
as the classical limit. Hence the transposed Poisson algebras that
can be quantized in this sense should be limited to the
aforementioned subclass of transposed Poisson algebras.
Unfortunately, the explicit characterization of this subclass in
terms of the (additional) identities has not been fixed yet, which
is still an open question.

Nevertheless, there are natural quantizations for a typical type
of transposed Poisson algebras, the so-called Novikov-Poisson
type, which are the ``commutators" of Novikov-Poisson algebras in
the sense that both the transposed Poisson algebras and the
Novikov-Poisson algebras have the same commutative associative
algebras and the Lie brackets in the former are the commutators of
the Novikov algebras in the latter. Note that transposed Poisson algebras in which the commutative associative
algebras are unital are of Novikov-Poisson type (\cite{BBGW-TPP}).
Recall that the notion of Novikov-Poisson algebras was
 introduced to construct a
tensor product theory of Novikov algebras as well as relate to
infinite-dimensional Lie algebras (\cite{Xu-NP}). We observe that
there is a natural Novikov deformation of the commutative
associative algebra in a Novikov-Poisson algebra. Hence the
relationship between Novikov-Poisson algebras and transposed
Poisson algebras given in \cite{BBGW-TPP} is revisited as an
immediate consequence by considering the corresponding classical
limits of such Novikov deformations. As a byproduct, all
transposed Poisson algebras of Novikov-Poisson type can be quantized.

Finally, we also classify quantizations of $2$-dimensional complex
transposed Poisson algebras where the Lie brackets are
non-abelian up to equivalence as a guide for a further
interpretation and development of the quantization theory of
transposed Poisson algebras.

This paper is organized as follows. In Section~\ref{s2}, we introduce the notion of the Novikov deformation of a commutative associative algebra and show that the classical limit belongs to a subclass of transposed Poisson algebras. In Section~\ref{s3}, we
give the natural Novikov deformation of the commutative
associative algebra in every Novikov-Poisson algebra. As an application, the existence of quantizations of transposed Poisson algebras of Novikov-Poisson type, including unital transposed Poisson algebras, is established. In Section~\ref{s4}, the
quantizations of $2$-dimensional complex transposed Poisson algebras with non-abelian Lie brackets are classified up to equivalence.

\smallskip

\noindent {\bf Notations.} Throughout this paper, we fix the base
field  $\mathbb{K}$ of characteristic $0$ for vector spaces. We
denote the power series ring $\mathbb{K}[[h]]$ by
$\mathbb{K}_{h}$. We
 use $\mathbb{N}$ and $\mathbb{Z}^{+}$ to denote the sets of natural numbers and positive integers, respectively.
 Associative algebras in this paper need not be unital.
\section{Novikov deformations and the corresponding classical limits}\label{s2}\

We give the notions of Novikov deformations and classical limits
of the Novikov deformations of commutative associative algebras.
We show that such classical limits belong to a subclass of transposed Poisson algebras.

\begin{defil}\label{eqvnov}
Let $A$ be a vector space and $\circ$ be a $\mathbb{K}$-bilinear operation on
$A$. Then the following two conditions are equivalent.
\begin{enumerate}
    \item\label{defi1}For all $x,y,z\in A$,
    \begin{eqnarray}
 (x\circ y)\circ z-(y\circ x)\circ z&=&x\circ(y\circ z)-y\circ(x\circ z),\label{no1} \\
  (x\circ y)\circ z&=&(x\circ z)\circ y\label{no2}.
\end{eqnarray}
\item\label{defi2}For all $x,y,z\in A$, {\rm (\ref{no2})} and the following equation hold.
    \begin{equation}\label{nctpa}
   2[x,y]\circ z=[x\circ z,y]+[x,y \circ z],
\end{equation}
where the operation $[,]$ on $A$ is defined by
\begin{equation}\label{eq:Lie}
[x,y]:=x\circ y-y\circ x,\;\;\forall x,y\in A.
\end{equation}
\end{enumerate}
The pair $(A,\circ)$ is called a \textbf{Novikov algebra} if it satisfies either of the two conditions above.
\end{defil}
\begin{proof}
By (\ref{no2}), we have
    \begin{equation}\label{rgd}
       (x\circ y)\circ z-(y\circ x)\circ z= (x\circ z)\circ y-(y\circ z)\circ x,\;\forall x,y,z\in A.
    \end{equation}

 ${\rm (\ref{defi1})}\Longrightarrow {\rm (\ref{defi2})}$.
We obtain (\ref{nctpa}) by adding
  (\ref{no1}) and (\ref{rgd}) together.

${\rm (\ref{defi2})}\Longrightarrow {\rm (\ref{defi1})}$.
 We obtain (\ref{no1}) by subtracting (\ref{rgd}) from (\ref{nctpa}).
\end{proof}

\begin{rmk}\label{rmk:ring}
In fact, the above conclusion and notion are still available if
the field $\mathbb K$ is replaced by a commutative ring $R$, that
is, $A$ is a module over a commutative ring $R$ and $\circ$ is a
$R$-bilinear operation on $A$. Accordingly, in this case,
$(A,\circ)$ is called a \textbf{Novikov algebra over $R$}.

\end{rmk}

The following two results are known (cf. \cite{Bai}).
\begin{lem}
A Novikov algebra is commutative if and only if it is a commutative associative algebra.
\end{lem}
\begin{defil}
\label{NL} Let $(A,\circ)$ be a Novikov algebra. Then the
commutator of $(A,\circ)$ defined by \rm{(\ref{eq:Lie})} makes
$(A,[,])$ a Lie algebra, which is called {\bf the sub-adjacent Lie
algebra of $(A,\circ)$}.
\end{defil}

\begin{rmk}\label{rmk:ring2}
As in Remark~\ref{rmk:ring}, the above conclusion still holds when
the field $\mathbb K$ is replaced by a commutative ring $R$, that
is, the commutator of a Novikov algebra over a commutative ring
$R$ defines a Lie algebra over $R$.
\end{rmk}

The following typical example of a Novikov algebra was given by S.Gel'fand.
\begin{ex}{\rm{(\cite{GD})}}\label{ex-dn}
Let $(A,\cdot)$ be a commutative associative algebra and $D$ be a derivation. Set
\begin{equation}\label{eq-dn}
x\circ y=x\cdot D(y),\quad \forall x,y\in A.
 \end{equation}
Then $(A,\circ)$ is a Novikov algebra.
\end{ex}
\begin{defi}{\rm{(\cite{BBGW-TPP})}}
    A \textbf{transposed Poisson algebra} is a triple $(A,\cdot,[,])$ consisting of a commutative associative algebra $(A,\cdot)$ and a Lie algebra $(A,[,])$ that satisfy the following
    compatible condition
    \begin{equation}\label{ttp}
2[x,y]\cdot z=[x,y\cdot z]+[x\cdot z,y],\quad\forall x,y,z\in A.
    \end{equation}
\end{defi}
Let $A$ be a vector space. Consider the set of formal series defined by
\[A_h := \Big\{\sum_{i=0}^{\infty}x_{i}~h^{i}~\Big|~x_{i}\in A\Big\}.\]
For
$x_{h}=\sum\limits_{i=0}^{\infty}x_{i}~h^{i},y_{h}=\sum\limits_{i=0}^{\infty}y_{i}~h^{i}\in
A_h$ and
$\lambda_{h}=\sum\limits_{i=0}^{\infty}\lambda_{i}~h^{i}\in\mathbb{K}_{h}$,
set
\[x_{h}+y_{h}:=\sum_{i=0}^{\infty}(x_{i}+y_{i})h^{i} , \quad \lambda_{h} x_{h}:=\sum_{i=0}^{\infty}\Big(\sum_{m+n=i}\lambda_{m}~x_{n}\Big)~h^{i}.\]
Then $A_h$ becomes a $\mathbb{K}_{h}$-module. Actually, $A_h$ is a
topologically free $\mathbb{K}_{h}$-module with respect to the
$h$-adic topology (see \cite{ES} for more details on
topologically free $\mathbb{K}_{h}$-modules). Moreover, if $A$ is
a vector space in dimension $n\in\mathbb{Z}^{+}$, then $A_h$ is a
finitely generated free $\mathbb{K}_{h}$-module of rank $n$.
\begin{defi}
  Let $(A,\circ)$ be a Novikov algebra. If $A_h$ is endowed with a $\mathbb{K}_{h}$-bilinear operation $\circ_{h}$
  such that $(A_h,\circ_{h})$ is a Novikov algebra over $\mathbb{K}_{h}$ and
  \begin{equation}
    x\circ_{h}y\equiv x\circ y\pmod h\label{modh}
  \end{equation}
  for all $x,y\in A$, then we call $(A_h,\circ_{h})$ a \textbf{Novikov deformation} of $(A,\circ)$. Two Novikov deformations $(A_h,\circ_{h})$
  and $(A_{h},\circ_{h}^{\prime})$ of $(A,\circ)$ are \textbf{equivalent} if there is a $\mathbb{K}_{h}$-linear
  isomorphism $f_{h}\in\mathrm{End}_{\mathbb{K}_{h}}(A_h)$ such that
  \begin{eqnarray*}
    f_{h}(x)&\equiv& x\pmod h,\\
    f_{h}(x\circ_{h}y)&=&f_{h}(x)\circ_{h}^{\prime}f_{h}(y),\forall x,y\in A.
  \end{eqnarray*}
\end{defi}

\begin{defi} Let $(A,\cdot)$ be a commutative associative algebra.
Suppose that $(\A,\cdot_h )$ is a Novikov deformation of
$(A,\cdot)$. Let $[,]$ be the $\mathbb K$-bilinear operation on
$A$ such that
  \begin{equation}
    [x,y]\equiv \frac{x\cdot_{h}y-y\cdot_{h}x}{h}\pmod h.\label{lim}
  \end{equation}
  Then the triple $(A,\cdot,[,])$ is called the \textbf{classical limit} of the Novikov deformation $(A_{h},\cdot_{h})$, and conversely, $(A_h,\cdot_{h})$ is called a \textbf{quantization} of $(A,\cdot,[,])$.
\end{defi}
\begin{thm}\label{thmnt}
  Let $(A,\cdot)$ be a commutative associative algebra and $(\A,\cdot_{h})$ be a Novikov deformation. Then the corresponding classical limit $(A,\cdot,[,])$ is a transposed Poisson algebra.
\end{thm}
\begin{proof}
Set
\begin{equation}\label{brah}
\{x_h,y_h\}_{h}:=x_h\cdot_{h}y_h-y_h\cdot_{h}x_h,\; \forall x_h,y_h\in A_h.
\end{equation}
Then
\begin{equation}
\{x,y\}_{h}\equiv [x,y]~h\pmod{h^2},\;\forall x,y\in A.
\end{equation}
Since $A_h$ is a topologically free $\mathbb{K}_h$-module and
$(A_h,\cdot_h)$ is a Novikov algebra over $\mathbb{K}_h$, by
Definition-Lemma~\ref{NL} and Remark~\ref{rmk:ring2}, we get
\begin{equation}
\{\{x,y\}_{h},z\}_{h}+\{\{y,z\}_{h},x\}_{h}+\{\{z,x\}_{h},y\}_{h}=0,\;\forall
x,y,z\in A.\label{ja}\end{equation} Taking the $h^{2}$ terms in
(\ref{ja}), we show that $(A,[,])$ is a Lie algebra. \delete{ Let
$x,y,z\in A$. By definition, we have
\begin{eqnarray}
(x\cdot_{h}y)\cdot_{h}z-(y\cdot_{h}x)\cdot_{h}z&=&x\cdot_{h}(y\cdot_{h}z)-y\cdot_{h}(x\cdot_{h}z),\label{n1} \\
(x\cdot_{h}y)\cdot_{h}z&=&(x\cdot_{h}z)\cdot_{h}y.\label{n2}
\end{eqnarray}
By (\ref{n2}), we have
\begin{equation}\label{guodu1}
     (x\cdot_{h}y)\cdot_{h}z-(y\cdot_{h}x)\cdot_{h}z=(x\cdot_{h}z)\cdot_{h}y-(y\cdot_{h}z)\cdot_{h}x.
\end{equation}
Adding (\ref{n1}) and (\ref{guodu1}) together, we get
\begin{equation}\label{guodu2}
    2((x\cdot_{h}y)\cdot_{h}z-(y\cdot_{h}x)\cdot_{h}z)=(x\cdot_{h}(y\cdot_{h}z)-(y\cdot_{h}z)\cdot_{h}x)+((x\cdot_{h}z)\cdot_{h}y-y\cdot_{h}(x\cdot_{h}z)),
\end{equation}
that is,}

Due to the fact that $A_h$ is a topologically free
$\mathbb{K}_h$-module and $(A_h,\cdot_h)$ is a Novikov algebra
over $\mathbb{K}_h$ again, by Definition-Lemma~\ref{eqvnov} and
Remark~\ref{rmk:ring}, we have
\begin{equation}\label{guodu3}
2\{x,y\}_{h}\cdot_h z=\{x,y\cdot_{h}z\}_{h}+\{x\cdot_{h}z,y\}_{h},\forall x,y,z\in A.
\end{equation}
Taking the first order terms of $h$ in (\ref{guodu3}), we obtain
(\ref{ttp}). Thus the conclusion follows.
\end{proof}

\begin{defi}
 Let $(A,\cdot,[,])$ be a transposed Poisson algebra.
    If there is a Novikov deformation $(A_h,\circ_h)$
    of $(A,\cdot)$ whose classical limit coincides with
    $(A,\cdot,[,])$, then we say that $(A,\cdot,[,])$ can be
    \textbf{quantized}. Otherwise we say that $(A,\cdot,[,])$
    \textbf{cannot be quantized}.
\end{defi}

    Let $\mathsf{P}$ be an operad (see \cite{Do-op} and \cite{LV-op} for more details on algebraic operads). For $n\in\mathbb{N}$, denote the vector space consisting of elements of arity $n$ in the operad $\mathsf{P}$ by $\mathsf{P}(n)$. Denote the operads of Novikov algebras and transposed Poisson algebras by $\mathsf{Nov}$ and $\mathsf{TPois}$ respectively.
Let $\mathsf{I}$ be the operadic ideal of $\mathsf{Nov}$ generated by its commutator \[[\mathsf{x_1},\mathsf{x_2}]=\mathsf{x_1}\circ\mathsf{x_2}-\mathsf{x_2}\circ\mathsf{x_1},\]
where $\mathsf{x_1},\mathsf{x_2}$ are placeholders.
Denote the associated graded operad of $\mathsf{Nov}$ with respect to the filtration by powers of the ideal $\mathsf{I}$ by $\mathsf{grNov}$, that is,
$$\mathsf{grNov}:=\mathsf{Nov}/F_\mathsf{I}\mathsf{Nov}\bigoplus_{k\ge 1}F_\mathsf{I}^k \mathsf{Nov}/F_\mathsf{I}^{k+1} \mathsf{Nov},$$
where the object $F_\mathsf{I}^k \mathsf{Nov}$ is spanned by operad compositions where elements of
$\mathsf{I}$ are used at least $k$ times (see \cite{Do-fm} and \cite{Do-ap} for more details on filtration on an operad by powers of an ideal).
We have
$$\mathrm{dim}\;\mathsf{grNov}(n)=\mathrm{dim}\;\mathsf{Nov}(n),\;\forall n\in\mathbb{N}.$$

\begin{pro}\label{proop}
Up to arity 4, the space of relations in $\mathsf{grNov}$ is
generated by the defining relations of $\mathsf{TPois}$.
\end{pro}
\begin{proof}
By the polarization method (\cite{MR-pol}),  $\mathsf{grNov}$ is
isomorphic to the operad encoding the classical limits. Thus, by
Theorem~\ref{thmnt}, the space of relations in $\mathsf{grNov}$
contains the relations generated by the defining relations of the
operad $\mathsf{TPois}$. By \cite[Theorem~1.1]{D-Novdim},
$\mathrm{dim}\;\mathsf{Nov}(n)=\binom{2n-2}{n-1}$ where
$n\in\mathbb{N}$. Up to arity $5$, the dimensions of
$\mathsf{TPois}(n)$ are $1$, $2$, $6$, $20$, $74$  by
\cite{Sar-cgd}. Note that for $n\leq 4$,
$$\mathrm{dim}\;\mathsf{grNov}(n)=\mathrm{dim}\;\mathsf{Nov}(n)=\mathrm{dim}\;\mathsf{TPois}(n).$$
Then the conclusion is obtained.
\end{proof}
 Thus, by Proposition~{\ref{proop}}, the classical limit $(A,\cdot,[,])$ satisfies no identities independent to the defining relations of the transposed Poisson algebra up to arity $4$. But $\mathrm{dim}\;\mathsf{Nov}(5)=70< \mathrm{dim}\;\mathsf{TPois}(5)$, and then there are more identities in $(A,\cdot,[,])$.
For example, by \cite{D-Novsp}, we obtain
    \begin{equation}\label{s5}
    \sum_{\sigma\in S_4}(-1)^{\sigma}[x_{\sigma(1)},[x_{\sigma(2)},[x_{\sigma(3)},[x_{\sigma(4)},x_{5}]]]]=0,\;\forall x_{i}\in A,\;1\leq i\leq 5.
    \end{equation}
 So the classical limit $(A,\cdot,[,])$ belongs to a subclass of transposed Poisson algebras and the operad encoding this subclass is isomorphic to the operad $\mathsf{grNov}$.
 It is still an open problem to give the exactly additional identities to characterize this subclass.
 Nevertheless, we introduce the notion as follows. If a transposed Poisson algebra belongs to this subclass, then it is called \textbf{a transposed Poisson algebra of quantizable type}.
 Obviously, if a transposed Poisson algebra is not of quantizable type, then it cannot be quantized.

We first study quantizations of the ``degenerate" transposed
Poisson algebras, that is, transposed Poisson algebras in which
the Lie algebras are abelian or the associative algebras are
trivial in the sense that the products are zero.  A transposed
Poisson algebra in which the Lie algebra is abelian can be quantized. For example, for a transposed Poisson algebra
$(A,\cdot,[,])$ in which $(A,[,])$ is abelian, the null
deformation of the commutative associative algebra $(A,\cdot)$ in
the sense that the operation $\cdot$ is extended
$\mathbb{K}_{h}$-bilinearly to $A_h$, giving a quantization of
$(A,\cdot,[,])$. On the other hand, note that for any Lie algebra
$(A,[,])$, $(A,\cdot, [,])$ is always a transposed Poisson algebra
in which the commutative associative algebra $(A,\cdot)$ is
trivial. Therefore, for the quantization of a transposed Poisson
algebra in which the commutative associative algebra is trivial,
we have the following conclusions.
\begin{pro}\label{pro-ql} Let $(A,[,])$ be a Lie algebra and $\cdot$ be a trivial product on $A$.
\begin{enumerate}
\item\label{it:aa} If $(A, [,])$ is the sub-adjacent Lie algebra
of a Novikov algebra, then the transposed Poisson algebra
$(A,\cdot, [,])$ can be quantized, which is of quantizable type
automatically.

\item \label{it:bb} If $(A, [,])$ is not the sub-adjacent Lie
algebra of a Novikov algebra, then $(A,\cdot, [,])$ cannot be
quantized. In particular, if $(A,[,])$ is a finite-dimensional
semisimple Lie algebra, then $(A,\cdot, [,])$ cannot be quantized.
\end{enumerate}
\end{pro}

\begin{proof}
(\ref{it:aa}). Let $(A,\circ)$ be a Novikov algebra whose
sub-adjacent Lie algebra is $(A, [,])$. Set
\begin{equation}\label{eq-dnp}
x\cdot_{h}y=x\cdot y+x\circ y~h,\quad \forall x,y\in A.
\end{equation}
Since $\cdot$ is a trivial product on $A$, we have
\begin{equation}\label{eq-dnp2}
x\cdot_{h}y=x\circ y~h,\quad \forall x,y\in A.
\end{equation}
Extend $\mathbb{K}_{h}$-bilinearly the above operation to $A_h$.
Note that
\begin{eqnarray*}
    (x\cdot_{h}y)\cdot_{h}z&=&((x\circ y)\circ z)~h^2,\\
    x\cdot_{h}(y\cdot_{h}z)&=&(x\circ (y\circ z))~h^2,\;\forall x,y,z\in A.
\end{eqnarray*}
Thus by (\ref{no1}) and (\ref{no2}), for all $x,y,z\in A$, we obtain
\begin{eqnarray}
(x\cdot_{h}y)\cdot_{h}z-(y\cdot_{h}x)\cdot_{h}z&=&x\cdot_{h}(y\cdot_{h}z)-y\cdot_{h}(x\cdot_{h}z)\; \textrm{and}\label{n1}\\
(x\cdot_{h}y)\cdot_{h}z&=&(x\cdot_{h}z)\cdot_{h}y,\label{n2}
\end{eqnarray}
 respectively.
Then $(\A,\cdot_{h})$ is a Novikov deformation of $(A,\cdot)$.
Since for all $x,y\in A$,
\begin{eqnarray*}
    \frac{x\cdot_{h}y-y\cdot_{h}x}{h}&=&x\circ y-y\circ x
    =[x,y],
\end{eqnarray*}
we conclude that $(\A,\cdot_{h})$ is a quantization of $(A,\cdot,[,])$.

(\ref{it:bb}). Suppose that $(A_h,\cdot_h)$ is a quantization of
the transposed Poisson algebra $(A,\cdot,[,])$. Let $\circ$ be a
$\mathbb K$-bilinear operation on $A$ that satisfies
    \[x\cdot_{h}y\equiv x\cdot y+ x\circ y~h\pmod{h^2},\;\forall x,y\in A.\]
    Since $\cdot$ is a trivial product on $A$, we obtain
    \[x\cdot_{h}y\equiv x\circ y~h\pmod{h^2},\;\forall x,y\in A.\]
    Then take the $h^{2}$ terms in (\ref{n1}) and (\ref{n2}), we show that $(A,\circ)$ is a Novikov algebra. By (\ref{lim}), we have
    \[x\circ y-y\circ x=[x,y],\;\forall x,y\in A,\]
which is a contradiction. The last conclusion follows from the
known fact (\cite{Bu-prielie}) that the sub-adjacent Lie algebra
of a finite-dimensional Novikov algebra over a field of
characteristic zero cannot be semisimple.
\end{proof}

As an application of the above conclusions, we illustrate that
there is a transposed Poisson algebra of quantizable type that
cannot be quantized.
\begin{ex}
There exists a transposed Poisson algebra $(A,\cdot,[,])$ of
quantizable type in which $(A,\cdot)$ is trivial and $(A,[,])$ is
not a commutator of a Novikov algebra, giving an example of a
transposed Poisson algebra of quantizable type  that cannot be
quantized. Let $(\mathbb{K}[t],\diamond)$ be the polynomial
algebra in one variable $t$. Set
\begin{equation*}
x\circ y=x\diamond \frac{dy}{dt},\quad \forall x,y\in \mathbb{K}[t].
\end{equation*}
Then $(\mathbb{K}[t],\circ)$ is a Novikov algebra by Example~\ref{ex-dn}.
Let $(\mathbb{K}[t],[,])$ be the sub-adjacent Lie algebra of $(\mathbb{K}[t],\circ)$
and $\cdot$ be a trivial product on $\mathbb{K}[t]$. Thus
$(\mathbb{K}[t],\cdot,[,])$ is a transposed Poisson algebra of quantizable
type by Proposition~\ref{pro-ql}~(\ref{it:aa}). Let $A$ be the $3$-dimensional subspace of $\mathbb{K}[t]$ with the basis $\{1,2t,-t^{2}\}$. Note that
\[[2t,-t^{2}]=-2t^{2},\;[2t,1]=-2,\;[-t^{2},1]=2t.\]
Hence $(A,\cdot,[,])$ is a subalgebra of the transposed Poisson
algebra $(\mathbb{K}[t],\cdot,[,])$ and $(A,[,])$ is isomorphic to
the simple Lie algebra $\mathfrak{sl}(2,\mathbb{K})$. So
$(A,\cdot,[,])$ cannot be quantized by
Proposition~\ref{pro-ql}~(\ref{it:bb}). On the other hand, since
any subalgebra of a transposed Poisson algebra of quantizable type
satisfies the relations of the operad encoding transposed Poisson
algebras of quantizable type, a subalgebra of a transposed Poisson
algebra of quantizable type is still of quantizable type.
Therefore, $(A,\cdot,[,])$ is a transposed Poisson algebra of
quantizable type, but it cannot be quantized.
\end{ex}

\section{Quantizations of transposed Poisson algebras of Novikov-Poisson type}\label{s3}

There is a natural Novikov deformation of the commutative
associative algebra in a Novikov-Poisson algebra. Hence the
relationship between Novikov-Poisson algebras and transposed
Poisson algebras given in \cite{BBGW-TPP} is revisited as an
immediate consequence. On the other hand, all transposed Poisson
algebras of Novikov-Poisson type, including unital transposed
Poisson algebras, can be quantized.

\delete{Let $(A,\circ)$ be a Novikov algebra
and $\circ_{h}$ a bilinear operation on $A[[h]]$. We exploit the
following series expansion of $\circ_{h}$
\begin{eqnarray}\label{expa}
  x\circ_{h}y=\sum_{n=0}^{\infty} x\circ_{n} y~h^{n}
\end{eqnarray}
where $x,y, x\circ_{n} y\in A$. It is clear that $(A[[h]],\circ_{h})$ is a Novikov deformation of $(A,\circ)$ if and only if
the following conditions are satisfied
\begin{enumerate}
  \item $x\circ_{0}y=x\circ y$,\label{1}
  \item $\underset{i+j=n}{\sum} ((x\circ_{i}y)\circ_{j}z-(y\circ_{i}x)\circ_{j}z)=\underset{i+j=n}{\sum} (x\circ_{i}(y\circ_{j}z)-y\circ_{i}(x\circ_{j}z))$,\label{2}
  \item $\underset{i+j=n}{\sum}(x\circ_{i}y)\circ_{j}z=\underset{i+j=n}{\sum}(x\circ_{i}z)\circ_{j}y$.\label{3}
\end{enumerate}
for all $x,y,z\in A, n\in \mathbb{N}$.}

\begin{defi}{\rm{(\cite{Xu-NP})}}
    A \textbf{Novikov-Poisson algebra} is a triple $(A,\cdot,\circ)$ where
    $(A,\cdot)$ is a commutative associative algebra and $(A,\circ)$ is a
    Novikov algebra such that
    \begin{eqnarray}
        (x\cdot y)\circ z &=& x\cdot(y\circ z),\label{np1}\\
        (x\circ y)\cdot z-(y\circ x)\cdot z &=& x\circ (y\cdot z)-y\circ (x\cdot z),\quad\forall x,y,z\in A\label{np2}.
    \end{eqnarray}
\end{defi}

\begin{ex}\label{ex-np} {\rm{(\cite{Xu-NP})}}
    Let $(A,\cdot)$ be a commutative associative algebra and $D$ be a derivation. Define a $\mathbb K$-bilinear operation $\circ$ on $A$ by (\ref{eq-dn}). Then $(A,\cdot,\circ)$ is a Novikov-Poisson algebra.
\end{ex}

\begin{pro}\label{pronp}
    Let $(A,\cdot,\circ)$ be a Novikov-Poisson algebra. Set
    \begin{equation}\label{eq-dnp}
    x\cdot_{h}y=x\cdot y+x\circ y~h,\quad \forall x,y\in A.
    \end{equation}
    Extend $\mathbb{K}_{h}$-bilinearly the above operation to $A_h$.  Then  $(\A,\cdot_{h})$ is a Novikov deformation of $(A,\cdot)$.
    Moreover, in its classical limit $(A,\cdot,[,])$,
    \begin{equation}\label{eq-conp}
    [x,y]=x\circ y-y\circ x,\quad \forall x,y\in A.
    \end{equation}
\end{pro}
\begin{proof}
  Let $x,y,z\in A$. Then we have
    \begin{eqnarray}
    (x\cdot_{h}y)\cdot_{h}z&=&(x\cdot y+x\circ y~h)\cdot_{h}z\nonumber\\
    &=& x\cdot y\cdot z+(x\circ y)\cdot z~h+(x\cdot y)\circ z~h+(x\circ y)\circ z~h^{2}\label{npd1}.
    \end{eqnarray}
    Exchanging $y$ and $z$ in (\ref{npd1}), we deduce
        \begin{eqnarray*}
    (x\cdot_{h}z)\cdot_{h}y&=&x\cdot z\cdot y+(x\circ z)\cdot y~h+(x\cdot z)\circ y~h+(x\circ z)\circ y~h^{2}\\
    &=&x\cdot y\cdot z+y\cdot(x\circ z)~h+(z\cdot x)\circ y~h+(x\circ z)\circ y~h^{2}\\
    &\overset{(\ref{np1})}{=}&x\cdot y\cdot z+(y\cdot x)\circ z~h+z\cdot(x\circ y)~h+(x\circ z)\circ y~h^{2}\\
    &\overset{(\ref{no2})}{=}&x\cdot y\cdot z+(y\cdot x)\circ z~h+z\cdot(x\circ y)~h+(x\circ y)\circ z~h^{2}\\
    &=& (x\cdot_{h}y)\cdot_{h}z.
    \end{eqnarray*}
Moreover, we have
    \begin{eqnarray}
    x\cdot_{h}(y\cdot_{h}z)&=&x\cdot_{h}(y\cdot z+y\circ z~h)\nonumber\\
    &=& x\cdot y\cdot z+x\cdot (y\circ z)~h+x\circ (y\cdot z)~h+ x\circ (y\circ z)~h^{2}
    \label{npd2}.
    \end{eqnarray}
    Comparing the difference between $(\ref{npd1})$ and $(\ref{npd2})$, we obtain
    \begin{equation}
(x\cdot_{h}y)\cdot_{h}z-x\cdot_{h}(y\cdot_{h}z)=((x\circ y)\cdot z-x\circ (y\cdot z))~h+((x\circ y)\circ z-x\circ (y\circ z))~h^{2}
\label{2ls}.
    \end{equation}
    Exchanging $x$ and $y$ in (\ref{2ls}), we have
    \begin{eqnarray*}
    (y\cdot_{h}x)\cdot_{h}z-y\cdot_{h}(x\cdot_{h}z)&=&((y\circ x)\cdot z-y\circ (x\cdot z))~h+((y\circ x)\circ z-y\circ (x\circ z))~h^{2}\\
    &\overset{(\ref{np2})}{=}&((x\circ y)\cdot z-x\circ (y\cdot z))~h+((y\circ x)\circ z-y\circ (x\circ z))~h^{2}\\
    &\overset{(\ref{no1})}{=}&((x\circ y)\cdot z-x\circ (y\cdot z))~h+((x\circ y)\circ z-x\circ (y\circ z))~h^{2}\\
    &=& (x\cdot_{h}y)\cdot_{h}z-x\cdot_{h}(y\cdot_{h}z).
\end{eqnarray*}
    Thus $(\A,\cdot_{h})$ is a Novikov algebra over $\mathbb{K}_h$.
    Then $(\A,\cdot_{h})$ is a Novikov deformation of $(A,\cdot)$ because
    $$ x\cdot_{h}y=x\cdot y+x\circ y~h\equiv x\cdot y\pmod h ,\;\forall x,y\in A.$$
    Also note that
    $$x\cdot_{h}y-y\cdot_{h}x=(x\circ y-y\circ x)~h,\quad \forall x,y\in A.$$
    Hence the conclusion holds.
\end{proof}

\begin{rmk}
For a Novikov-Poisson algebra $(A,\cdot,\circ)$, there is a
Novikov deformation of the Novikov algebra $(A,\circ)$ (not the
commutative associative algebra $(A,\cdot)$) given in
\cite{ZBM-NP}. Explicitly, set
 \begin{equation*}
 x\circ_{h}y=x\circ y+x\cdot y~h,\quad\forall x,y\in A.
 \end{equation*}
 Extend $\mathbb{K}_{h}$-bilinearly the above operation to $A_h$
 and then $(\A,\circ_{h})$ is a Novikov deformation of $(A,\circ)$.
\end{rmk}

In \cite{BBGW-TPP}, the following result is proved by a direct
proof, which is revisited as follows.
\begin{cor}{\rm(\cite{BBGW-TPP})}
    Let $(A,\cdot,\circ)$ be a Novikov-Poisson algebra. Define a $\mathbb K$-bilinear operation $[,]$ on $A$ by $(\ref{eq-conp})$. Then $(A,\cdot,[,])$ is a transposed Poisson algebra.
\end{cor}
\begin{proof}
The classical limit of the Novikov deformation of $(A,\cdot)$ given in Proposition \ref{pronp} is exactly $(A,\cdot,[,])$. Thus the conclusion follows from Theorem~\ref{thmnt}.
\end{proof}
\begin{rmk}
    We call $(A,\cdot,[,])$ \textbf{the sub-adjacent transposed Poisson algebra of} $(A,\cdot,\circ)$.
\end{rmk}

\begin{ex}\label{ex-dd}
\label{prod}
    Let $(A,\cdot)$ be a commutative associative algebra and $D$ be a derivation.
With the $\mathbb K$-bilinear operation $\circ$ on $A$ defined by
(\ref{eq-dn}), $(A,\cdot,\circ)$ is a Novikov-Poisson algebra (see
Example \ref{ex-np}). In this case, (\ref{eq-dnp}) takes the form
\begin{equation}x\cdot_{h}y= x\cdot y+ x\circ y~h=x\cdot y+x\cdot D(y)~h,\;\forall x,y\in
A,
\end{equation}
giving a Novikov deformation $(A_h,\cdot_{h})$ of $(A,\cdot)$ by
extending $\mathbb{K}_{h}$-bilinearly the above operation to $A_h$. The classical limit $(A,\cdot,[,])$ in which the Lie bracket is defined by
\begin{equation}\label{eq-cd}
    [x,y]=x\cdot D(y)-y\cdot D(x),\quad \forall x,y\in A,
    \end{equation}
is exactly the typical example of a transposed Poisson algebra given in \cite{BBGW-TPP}.
\end{ex}


\begin{defi}\label{def:NP} Let $(A,\cdot,[,])$ be  a
   transposed Poisson algebra.
    If there is a Novikov-Poisson
   algebra $(A,\cdot,\circ)$ such that $(A,\cdot,[,])$ is the sub-adjacent transposed Poisson algebra of
   $(A,\cdot,\circ)$, then $(A,\cdot,[,])$ is called \textbf{a transposed Poisson algebra  of Novikov-Poisson
type}.

   \end{defi}

\begin{pro}\label{pro-qstpa}
 Any transposed Poisson algebra of Novikov-Poisson type can be quantized.
\end{pro}
\begin{proof}
        Let $(A,\cdot,[,])$ be a transposed Poisson algebra of Novikov-Poisson type. By Definition~\ref{def:NP}, there is a Novikov-Poisson algebra $(A,\cdot,\circ)$ satisfying
        \[  [x,y]=x\circ y-y\circ x,\quad \forall x,y\in A.\]
 Define $\cdot_{h}$ by (\ref{eq-dnp}) and extend $\mathbb{K}_{h}$-bilinearly it to $A_h$.
 Thus, by Proposition~\ref{pronp}, $(\A,\cdot_{h})$ is a quantization of $(A,\cdot,[,])$.
\end{proof}
\begin{rmk}\label{rmk-gdg}
    Let $(A,\circ)$ be a Novikov algebra and $(A,[,])$ be the sub-adjacent Lie algebra of $(A,\circ)$. Then $(A,\cdot,\circ)$ with $\cdot$ being trivial is a Novikov-Poisson algebra,
    and thus $(A,\cdot,[,])$ is a transposed Poisson algebra of
Novikov-Poisson type. So $(A,\cdot,[,])$ can be quantized. Thus
Proposition~\ref{pro-ql}~(\ref{it:aa}) is a special case
    of Proposition~\ref{pro-qstpa}.
\end{rmk}
\begin{cor}
    A unital transposed Poisson algebra $(A,\cdot,[,])$ in the sense that $(A,\cdot)$ is unital can be quantized.
\end{cor}

\begin{proof} It follows from the fact that any unital transposed Poisson algebra
$(A,\cdot,[,])$ is of Novikov-Poisson type (\cite{BBGW-TPP}).
Explicitly, denote the unit element of $(A,\cdot)$ by $1_{A}$ and
    \delete{set
    \begin{equation}\label{eq-qutpa}
    x\cdot_{h}y=x\cdot y+x\cdot [1_{A},y]~h,\quad\forall x,y\in A.
    \end{equation}
    Extend $\mathbf{K}$-bilinearly the above operation to $A[[h]]$.}
    define a linear map $D\in\mathrm{End}_{\mathbb{K}}(A)$ by
    \[D(x):=[1_{A},x],\forall x\in A.\]
     Then $D$ is a derivation of $(A,\cdot)$
    and $(\ref{eq-cd})$ holds. Thus $(A,\cdot,[,])$ is of Novikov-Poisson type.
\end{proof}
\section{Quantizations of 2-dimensional complex transposed Poisson algebras}\label{s4}
We classify quantizations of $2$-dimensional transposed Poisson
algebras where the Lie algebras are non-abelian up to equivalence.
In this section, we assume that the base field $\mathbb K$ is the
field $\mathbb C$ of complex numbers.

We have studied  the existence of quantizations of transposed
Poisson algebras in the previous sections. However, there might be non-equivalent quantizations of a transposed
Poisson algebra. So we
turn to the problem on the classification of the quantizations of
a transposed Poisson algebra up to equivalence.
As a guide, we consider the case of 2-dimensional
complex transposed Poisson algebras. 

We have shown that a transposed Poisson algebra in which the Lie
algebra is abelian can be quantized. But even for $2$-dimensional
transposed Poisson algebras with abelian Lie brackets, it seems
quite hard to classify their quantizations into mutually
non-equivalent classes. However, for $2$-dimensional transposed
Poisson algebras in which the Lie algebras are non-abelian, we
obtain a complete classification of their quantizations.

Let $(A,\cdot,[,])$ be a $2$-dimensional
transposed Poisson algebra in which the Lie algebra $(A,[,])$ is non-abelian. Let $\{e_{1},e_{2}\}$ be a basis of $A$. By the arXiv version of \cite{BBGW-TPP}, it is isomorphic to one of the following mutually non-isomorphic transposed Poisson algebras whose non-zero products are given by
\begin{enumerate}
  \item[$A^{0,0}$]: $[e_{1},e_{2}]=e_{2}$;\label{c1}
  \item[$A^{0,1}$]: $e_{1}\cdot e_{1}=e_{2},\quad[e_{1},e_{2}]=e_{2}$;\label{c2}
  \item[$A^{\lambda,0}$]:  $e_{1}\cdot e_{1}=\lambda e_{1},\; e_{1}\cdot e_{2}=e_{2}\cdot e_{1}=\lambda e_{2},\;0\neq\lambda\in \mathbb{C},\quad[e_{1},e_{2}]=e_{2}$.\label{c3}
\end{enumerate}

\begin{lem}\label{lem:cc} Let $(A,[,])$ be a 2-dimensional Lie algebra with a
basis $\{e_1,e_2\}$. Assume that
\begin{equation}\label{2com}
            [e_{1},e_{2}]=e_{2}.
        \end{equation}
If $(A,\circ)$ is a Novikov algebra whose sub-adjacent Lie algebra
is $(A,[,])$, then $(A,\circ)$ is exactly of one of the following
forms (all of them are Novikov algebras):
   \begin{equation}\label{eq:ss}e_{1}\circ e_{1}=ae_{1}+be_{2},\; e_{1}\circ e_{2}=(a+1)e_{2},\;e_{2}\circ e_{1}=ae_{2},\;e_{2}\circ e_{2}=0,\end{equation}
    where the parameters $a,b\in\mathbb{C}$. Moreover, for any commutative associative algebra $(A,\cdot)$ in one of the transposed Poisson
    algebras $A^{0,0}$, $A^{0,1}$ and $A^{\lambda,0}$ with $\lambda\ne 0$,
     $(A,\cdot,\circ)$ is a
    Novikov-Poisson algebra for all $a,b\in\mathbb{C}$.
\end{lem}

\begin{proof}
Let $(A,\circ)$ be a Novikov algebra whose sub-adjacent Lie algebra
is $(A,[,])$. Set $$e_{1}\circ e_{j}=\sum_{i=1}^{2}\alpha_{ij}e_{i},\;e_{2}\circ e_{j}=\sum_{i=1}^{2}\beta_{ij}e_{i},\;\alpha_{ij},\beta_{ij}\in\mathbb{C},\;1\leq i,j\leq 2.$$
 By (\ref{2com}), we have
 \begin{equation}\label{eq1}
      \alpha_{12}=\beta_{11},\;\alpha_{22}-\beta_{21}=1.
 \end{equation}
  Considering (\ref{nctpa}) with $x=e_1,y=e_2,z=e_1$, we have
 \begin{equation}\label{eq2}
\beta_{11}=0,\;\alpha_{11}=\beta_{21}.
\end{equation}
Considering (\ref{nctpa}) with $x=e_1,y=e_2,z=e_2$, we have
\begin{equation}\label{eq3}
 \beta_{12}=0,\;\alpha_{12}=\beta_{22}.
\end{equation}
Setting $a:=\beta_{21}$ and $b:=\alpha_{21}$, we obtain
(\ref{eq:ss}) by (\ref{eq1}), (\ref{eq2}) and (\ref{eq3}). It is
straightforward to show that the $\mathbb C$-bilinear operation
$\circ$ defined by (\ref{eq:ss}) is a Novikov algebra. Note that
in these cases, (\ref{no2}) holds automatically.

For any commutative associative algebra $(A,\cdot)$ in one of the transposed Poisson algebras $A^{0,0}$, $A^{0,1}$ and $A^{\lambda,0}$ with $\lambda\ne 0$, and any $a,b\in\mathbb{C}$, it is also straightforward to verify that (\ref{np1}) and (\ref{np2}) hold. Hence $(A,\cdot,\circ)$ is a Novikov-Poisson algebra.
\delete{
By definition, it remains to  show (\ref{np1}) and (\ref{np2}) hold in each case.
\begin{enumerate}
    \item $(A,\cdot,[,])=A^{0,0}$. \\
    Since $\cdot$ is a trivial operation,
    then (\ref{np1}) and (\ref{np2}) hold.
     \item $(A,\cdot,[,])=A^{0,1}$.\\
     We have
     $$(e_1\circ e_2)\cdot z-(e_2\circ e_1)\cdot z=0=e_1\circ(e_2\cdot z)-e_2\circ(e_1\cdot z),\;\forall z\in A.$$
    Hence we obtain (\ref{np2}). We obtain (\ref{np1}) because
    \begin{eqnarray*}
        (e_1\cdot e_2)\circ z&=&0=e_1\cdot(e_2\circ z),\\
         (e_1\cdot e_1)\circ e_1&=&ae_2=e_1\cdot(e_1\circ e_1),\\
          (e_1\cdot e_1)\circ e_2&=&0=e_1\cdot(e_1\circ e_2),\\
           (e_2\cdot y)\circ z&=&0=e_2\cdot(y\circ z),\;\forall y,z\in A.
    \end{eqnarray*}
 \item $(A,\cdot,[,])=A^{\lambda,0}$.\\
  In this case, we have
 \[e_1\cdot x=x\cdot e_1=\lambda x,\;\forall x\in A.\]
 Thus
 \[(e_1\cdot x)\circ y=\lambda x\circ y=e_1\cdot(x\circ y),\forall x,y\in A.\]
    Also note that
    \begin{eqnarray*}
         (e_2\cdot e_1)\circ e_1&=&a\lambda e_2=e_2\cdot(e_1\circ e_1),\\
          (e_2\cdot e_1)\circ e_2&=&0=e_2\cdot(e_1\circ e_2),\\
           (e_2\cdot e_2)\circ z&=&0=e_2\cdot(e_2\circ z),\;\forall x, y,z\in A,
    \end{eqnarray*}
    we get (\ref{np1}).
    We obtain (\ref{np2}) since
    \begin{eqnarray*}
        &&(e_1\circ e_2)\cdot e_1-(e_1\circ e_2)\circ e_1=\lambda e_2=e_1\circ(e_2\cdot e_1)-e_2\circ(e_1\cdot e_1),\\
        &&(e_1\circ e_2)\cdot e_2-(e_1\circ e_2)\cdot e_2=0=e_1\circ(e_2\cdot e_2)-e_2\circ(e_1\cdot e_2).
    \end{eqnarray*}
\end{enumerate}}
\end{proof}

\begin{rmk}
In fact, the classification of 2-dimensional complex Novikov algebras whose sub-adjacent Lie algebra are not abelian up to isomorphism was given in \cite{BM1}, that is, the Novikov algebras of the above forms have already been classified up to isomorphism.
\end{rmk}

\begin{cor}
The transposed Poisson
 algebras $A^{0,0}$, $A^{0,1}$ and $A^{\lambda,0}$ with $\lambda\ne 0$ are of Novikov-Poisson type. Hence all of them can be quantized.
\end{cor}

\begin{proof}
It follows immediately from Lemma~\ref{lem:cc} and
Proposition~\ref{pro-qstpa}.
\end{proof}

\delete{
\begin{rmk} Note that by \cite[Example~3.5]{BBGW-TPP}, $A^{0,1}$ and $A^{\lambda,0}$
have been shown to be transposed Poisson algebras of
Novikov-Poisson type, corresponding to the above case that
$a=?,b=?$ and $a=0, b=0$ respectively. \cm{please add}

\sy{In \cite[Example~3.5]{BBGW-TPP}, there is an error in case (5) corresponding to the associative algebra in $A^{0,1}$.
\begin{eqnarray*}
(5)\;(L,\cdot):e_1e_1=e_2,;\;\mathrm{Der}_a(L)= \left(
    \begin{matrix}
        a&0\\
    b&2a
    \end{matrix}
    \right);\; (L,[,]):[e_1,e_2]=be_2.
\end{eqnarray*}
We should have
\begin{eqnarray*}
    [e_1,e_2]&=&e_1D(e_2)-e_2D(e_1)\\
    &=&e_1(2ae_2)-e_2(ae_1+be_2)\\
    &=&0.
\end{eqnarray*}
 }
\end{rmk}
}

Our next task is the classification of their quantizations up to
equivalence.
\begin{lem}\label{lem-sf}
    Let $A$ be a $2$-dimensional vector space and $\{e_{1},e_{2}\}$ be a basis. If there is a $\mathbb{C}_{h}$-bilinear operation $\cdot_{h}$ on $A_h$ such that $(A_h,\cdot_{h})$ is a Novikov algebra over $\mathbb{C}_h$ and
    \begin{equation}\label{eq-comu}
        e_{1}\cdot_{h}e_{2}-e_{2}\cdot_{h}e_{1}\equiv e_{2}~h\pmod{h^2},
    \end{equation}
    then there are $(e_{1})_{h},(e_{2})_{h}\in A_h$ satisfying
    \begin{enumerate}
        \item[(\romannumeral1)] $(e_{1})_{h}\equiv e_{1}\pmod{h},\;(e_{2})_{h}\equiv e_{2}\pmod{h}$;
        \item[(\romannumeral2)] $(e_{1})_{h}\cdot_{h}(e_{1})_{h}=a_{h}(e_{1})_h+b_{h}(e_{2})_h,\;(e_{1})_{h}\cdot_{h}(e_{2})_{h}=(a_{h}+h)(e_{2})_h,\\(e_{2})_{h}\cdot_{h}(e_{1})_{h}= a_{h}(e_{2})_{h},\;(e_{2})_{h}\cdot_{h}(e_{2})_{h}=0$, where $a_{h},b_{h}\in \mathbb{C}_{h}$.
    \end{enumerate}
\end{lem}
\begin{proof}
    By (\ref{eq-comu}), we obtain
    \begin{equation}\label{eq-gd1}
e_{1}\cdot_{h}e_{2}-e_{2}\cdot_{h}e_{1}= h(\mu_{h}e_{1}+\nu_{h}e_{2}),
    \end{equation}
where $\mu_{h}$ and $\nu_{h}$ are elements of $\mathbb{C}_{h}$
 satisfying
    \begin{eqnarray*}
        \mu_{h}&\equiv& 0\pmod{h},\\
        \nu_{h}&\equiv& 1\pmod{h}.
    \end{eqnarray*}
 Thus $\nu_{h}$ is an invertible element in the ring $\mathbb{C}_{h}$. Define elements $(e_{1})_{h}$ and $(e_{2})_{h}$ in $A_h$ by
 \begin{eqnarray}
    (e_{1})_{h}&:=&\nu_{h}^{-1}e_{1},\label{b1}\\
(e_{2})_{h}&:=&\nu_{h}^{-1}\mu_{h}e_{1}+e_{2}\label{b2}.
 \end{eqnarray}
 So $(e_{1})_{h}$ and $(e_{2})_{h}$ satisfy Condition (\romannumeral1).

  By (\ref{eq-gd1}), we get
 \begin{equation}\label{eq-gd2}
        (e_{1})_{h}\cdot_{h}(e_{2})_{h}-(e_{2})_{h}\cdot_{h}(e_{1})_{h}= h(e_{2})_{h}.
        \end{equation}
Note that $\{e_{1}, e_{2}\}$ is a basis of the finitely generated
free $\mathbb{C}_h$-module $A_h$. Hence $\{(e_{1})_{h},
(e_{2})_{h}\}$ is also a basis of the finitely generated free
$\mathbb{C}_h$-module $A_h$ by (\ref{b1}) and (\ref{b2}). Set
 $$(e_{1})_{h}\cdot_{h}(e_{j})_{h}=\sum_{i=1}^{2}(\alpha_{ij})_{h}(e_{i})_{h},\;(e_{2})_{h}\cdot_{h}(e_{j})_{h}=\sum_{i=1}^{2}(\beta_{ij})_{h}(e_{i})_{h},\;(\alpha_{ij})_{h},(\beta_{ij})_{h}\in\mathbb{C}_{h},\;1\leq i,j\leq 2.$$
    By (\ref{eq-gd2}), we have
        \begin{equation}\label{eqh1}
            (\alpha_{12})_{h}=(\beta_{11})_{h},\;(\alpha_{22})_{h}-(\beta_{21})_{h}=h.
        \end{equation}
        Define a $\mathbb C_h$-bilinear operation $\{,\}_h$ on $A_h$ by (\ref{brah}).
        By 
Definition-Lemma~\ref{eqvnov} and Remark~\ref{rmk:ring}, we have
        \begin{equation}\label{guodu6}
            2\{x_h,y_h\}\cdot_h z_h=\{x_h\cdot_h z_h,y_h\}+\{x_h,y_h\circ_h z_h\},\;\forall x_h,y_h,z_h\in A_h.
        \end{equation}
        Similar to 
        the proof of Lemma~\ref{lem:cc},
        we can deduce from (\ref{guodu6}) that
        \begin{eqnarray}
          && (\beta_{11})_{h}=0,\;(\alpha_{11})_{h}=(\beta_{21})_{h},\label{eqh2}\\
            &&(\beta_{12})_{h}=0,\;(\alpha_{12})_{h}=(\beta_{22})_{h}.\label{eqh3}
        \end{eqnarray}
Set $a_{h}:=(\beta_{21})_{h}$ and $b_{h}:=(\alpha_{21})_{h}$, then by (\ref{eqh1}), (\ref{eqh2}) and (\ref{eqh3}) we have
\[(e_{1})_{h}\cdot_{h}(e_{1})_{h}=a_{h}(e_{1})_h+b_{h}(e_{2})_h,\;(e_{1})_{h}\cdot_{h}(e_{2})_{h}=(a_{h}+h)(e_{2})_h,\;(e_{2})_{h}\cdot_{h}(e_{1})_{h}= a_{h}(e_{2})_{h},\;(e_{2})_{h}\cdot_{h}(e_{2})_{h}=0.\]
Note that
\begin{eqnarray*}
    ((e_{1})_{h}\cdot_{h}(e_{1})_{h})\cdot_h(e_{2})_{h}&=&a_h(a_h+h)(e_{2})_{h}=((e_{1})_{h}\cdot_{h}(e_{2})_{h})\cdot_h(e_{1})_{h},\\
((e_{2})_{h}\cdot_{h}(e_{1})_{h})\cdot_h(e_{2})_{h}&=&0=((e_{2})_{h}\cdot_{h}(e_{2})_{h})\cdot_h(e_{1})_{h}.
\end{eqnarray*}
Thus the following equation automatically holds.
\[(x_h\cdot_h y_h)\cdot_h z_h=(x_h\cdot_h z_h)\cdot_h y_h,\;\forall x_h,y_h,z_h\in A_h.\]
Hence by 
Definition-Lemma~\ref{eqvnov} and Remark~\ref{rmk:ring} again, we
obtain Condition (\romannumeral2). \delete{
    For all $x_{h},y_{h}\in \A$, define $\mathbb{K}_{h}$-linear maps $L_{\cdot_{h}},R_{\cdot_{h}}: \A\to\mathrm{End}_{\mathbb{K}_{h}}(\A)$ by
    \[L_{\cdot_{h}}(x_h)y_h:=x_h\cdot_h y_h,\; R_{\cdot_{h}}(x_h)y_{h}:=y_h\cdot_h x_h. \]
    Then $(\A,\cdot_{h})$ is a Novikov algebra if and only if
    \begin{eqnarray}
        L_{\cdot_{h}}((e_{1})_{h})L_{\cdot_{h}}((e_{2})_{h})-L_{\cdot_{h}}((e_{2})_{h})L_{\cdot_{h}}((e_{1})_{h})&=&L_{\cdot_{h}}((e_{1})_{h}\cdot_{h}(e_{2})_{h}-(e_{2})_{h}\cdot_{h}(e_{1})_{h}),\label{g2n1}\\
        R_{\cdot_{h}}((e_{1})_{h})R_{\cdot_{h}}((e_{2})_{h})&=&R_{\cdot_{h}}((e_{2})_{h})R_{\cdot_{h}}((e_{1})_{h})\label{2n1}.
    \end{eqnarray}
    By (\ref{eq-gd2}), (\ref{g2n1}) reduces to
    \begin{equation}\label{2n2}
        L_{\cdot_{h}}((e_{1})_{h})L_{\cdot_{h}}((e_{2})_{h})-L_{\cdot_{h}}((e_{2})_{h})L_{\cdot_{h}}((e_{1})_{h})=hL_{\cdot_{h}}((e_{2})_{h}).
    \end{equation}
        Set $$(e_{1})_{h}\cdot_{h}(e_{j})_{h}=\sum_{i=1}^{2}(\alpha_{ij})_{h}(e_{i})_{h},\;(e_{2})_{h}\cdot_{h}(e_{j})_{h}=\sum_{i=1}^{2}(\beta_{ij})_{h}(e_{i})_{h},\;(\alpha_{ij})_{h},(\beta_{ij})_{h}\in\mathbb{K}_{h},\;1\leq i,j\leq 2.$$
        We deduce from (\ref{2n1}) and (\ref{2n2}) that
        \begin{eqnarray}
            \left(
            \begin{matrix}
            (\alpha_{11})_{h}&(\beta_{11})_{h}\\
            (\alpha_{21})_{h}&(\beta_{21})_{h}
            \end{matrix}
            \right)
            \left(
            \begin{matrix}
            (\alpha_{12})_{h}&(\beta_{12})_{h}\\
            (\alpha_{22})_{h}&(\beta_{22})_{h}
            \end{matrix}
            \right)=\left(
            \begin{matrix}
            (\alpha_{12})_{h}&(\beta_{12})_{h}\\
            (\alpha_{22})_{h}&(\beta_{22})_{h}
            \end{matrix}
            \right)
            \left(
            \begin{matrix}
            (\alpha_{11})_{h}&(\beta_{11})_{h}\\
            (\alpha_{21})_{h}&(\beta_{21})_{h}
            \end{matrix}
            \right)\label{eq1}
\end{eqnarray}
        and
\begin{eqnarray} \left(
            \begin{matrix}
            (\alpha_{11})_{h}&(\alpha_{12})_{h}\\
            (\alpha_{21})_{h}&(\alpha_{22})_{h}
            \end{matrix}
            \right)
            \left(
            \begin{matrix}
        (\beta_{11})_{h}&(\beta_{12})_{h}\\
        (\beta_{21})_{h}&(\beta_{22})_{h}
            \end{matrix}
            \right)-    \left(
            \begin{matrix}
            (\beta_{11})_{h}&(\beta_{12})_{h}\\
            (\beta_{21})_{h}&(\beta_{22})_{h}
            \end{matrix}
            \right)
                \left(
            \begin{matrix}
            (\alpha_{11})_{h}&(\alpha_{12})_{h}\\
            (\alpha_{21})_{h}&(\alpha_{22})_{h}
            \end{matrix}
            \right)=    \left(
            \begin{matrix}
            (\beta_{11})_{h}h&(\beta_{12})_{h}h\\
            (\beta_{21})_{h}h&(\beta_{22})_{h}h
            \end{matrix}
            \right)\label{eq2}
        \end{eqnarray}
        respectively.
        Solving the equations in (\ref{eq1}), (\ref{eq2}) and (\ref{eq3}), we
        deduce
        \begin{eqnarray*}
            (\alpha_{11})_{h}&=&(\beta_{21})_{h},\\
            (\alpha_{22})_{h}&=&(\beta_{21})_{h}+h,\\
            (\alpha_{12})_{h}&=&(\beta_{11})_{h}=(\beta_{12})_{h}=(\beta_{22})_{h}=0.
        \end{eqnarray*}
Hence Condition (\romannumeral2) follows by setting
$a_{h}:=(\beta_{21})_{h}$ and $b_{h}:=(\alpha_{21})_{h}$.}
\end{proof}
    Let $A$ be a $2$-dimensional vector space, $\{e_{1},e_{2}\}$ be a basis and $a_{h},b_{h}$ be elements in $\mathbb{C}_{h}$. Define a $\mathbb{C}_{h}$-bilinear operation $\cdot_h$ on $\A$ by
    \[e_{1}\cdot_{h}e_{1}=a_{h}e_{1}+b_{h}e_{2},\;e_{1}\cdot_{h}e_{2}=(a_{h}+h)e_{2},\;e_{2}\cdot_{h}e_{1}=a_{h}e_{2},\;e_{2}\cdot_{h}e_{2}=0.\]
     Note that (\ref{eq-gd1}) holds if we set
     $$\mu_{h}=0,\;\nu_{h}=1.$$ Moreover, in this case, we have
     $$(e_1)_h=e_1,\;(e_2)_h=e_2.$$
   Thus $(\A,\cdot_{h})$ is a Novikov algebra over $\mathbb{C}_h$ by the
   proof of Condition (\romannumeral2) in Lemma~\ref{lem-sf}.
    Since $(\A,\cdot_{h})$ is determined by parameters $a_{h}$ and $b_{h}$, we denote it by $A_{h}^{a_{h},b_{h}}$.
    \begin{pro}\label{eq-q}
        Let $(A,\cdot,[,])$ be a 2-dimensional transposed Poisson
        algebra with a basis $\{e_{1},e_{2}\}$. Assume that
        \begin{equation*}
            [e_{1},e_{2}]=e_{2}.
        \end{equation*}
        Then the following results hold.
        \begin{enumerate}
            \item Let $(A_h,\cdot_{h}^{''})$ be a quantization of $(A,\cdot,[,])$.
            Then there are elements $a_{h}$ and $b_{h}$ in $\mathbb{C}_{h}$ such that
            $A_{h}^{a_h,b_h}$ is a quantization of $(A,\cdot,[,])$, which is equivalent
            to $(A_h,\cdot_{h}^{''})$. \label{e1}
            \item Let $A_{h}^{a_h,b_h}$ and $A_{h}^{a^{\prime}_h, b^{\prime}_h}$ be quantizations of $(A,\cdot,[,])$, where $a_h,b_h,a^{\prime}_h, b^{\prime}_h\in\mathbb{C}_{h}$. Then $A_{h}^{a_h,b_h}$ and $A_{h}^{a^{\prime}_h, b^{\prime}_h}$ are equivalent if and only if
            the following conditions are satisfied.\label{e2}
            \begin{enumerate}
                \item $a_h=a^{\prime}_h$.\label{e2-1}
                \item There are $\epsilon_{h}\equiv 1\pmod{h},\mu_{h}\in \mathbb{C}_{h}$ such that\label{e2-2}
                \begin{equation}
                    b_{h}^{\prime}=b_{h}\epsilon_{h}-\mu_{h}h(a_{h}+h).\label{eqv}
                \end{equation}
            \end{enumerate}
        \end{enumerate}
    \end{pro}
\begin{proof}
    (\ref{e1}). Since $(A_h,\cdot_{h}^{''})$ is a quantization of $(A,\cdot,[,])$, we have
    \[\frac{e_{1}\cdot_{h}^{''}e_{2}-e_{2}\cdot_{h}^{''}e_{1}}{h}\equiv e_2\pmod h.\]
    Then by Lemma~\ref{lem-sf}, there are $(e_{1})_{h},(e_{2})_{h}\in A_h$ and $a_{h},b_{h}\in \mathbb{C}_{h}$ such that
    \begin{eqnarray*}
        (e_{1})_{h}&\equiv& e_{1}\pmod{h},\;(e_{2})_{h}\equiv e_{2}\pmod{h},\\
        (e_{1})_{h}\cdot_{h}^{''}(e_{1})_{h}&=&a_{h}(e_{1})_h+b_{h}(e_{2})_h,\;(e_{1})_{h}\cdot_{h}^{''}(e_{2})_{h}=(a_{h}+h)(e_{2})_h,
        \\(e_{2})_{h}\cdot_{h}^{''}(e_{1})_{h}&=& a_{h}(e_{2})_{h},\;(e_{2})_{h}\cdot_{h}^{''}(e_{2})_{h}=0.
    \end{eqnarray*}
    Thus $A_{h}^{a_h,b_h}$ is a quantization of $(A,\cdot,[,])$, which is equivalent
    to $(A_h,\cdot_{h}^{''})$.

(\ref{e2}).
    Assume that $(A_h,\cdot_h):=A_{h}^{a_{h},b_{h}}$ and $(A_h,\cdot^{\prime}_{h}):=A_{h}^{a_{h}^{\prime},b_{h}^{\prime}}$ are equivalent and $f_{h}: A_{h}^{a_{h},b_{h}}\to A_{h}^{a_{h}^{\prime},b_{h}^{\prime}}$ is a
 $\mathbb{C}_{h}$-linear
  isomorphism giving the equivalence between them. Since
    \[f_{h}(e_{1})\cdot_{h}^{\prime}f_{h}(e_{2})-f_{h}(e_{2})\cdot_{h}^{\prime}f_{h}(e_{1})=f_{h}(e_{1}\cdot_{h}e_{2}-e_{2}\cdot_{h}e_{1})=hf_{h}(e_{2}),\]
there are $\epsilon_{h}\equiv 1\pmod{h}, \mu_{h}\in
\mathbb{C}_{h}$ such that
\begin{equation}\label{df}
    f_{h}(e_{2})=\epsilon_{h}e_{2},\;
    f_{h}(e_{1})=e_{1}+\mu_{h}he_{2}.
\end{equation}
Thus we have \begin{eqnarray*}
    f_{h}(e_{1})\cdot_{h}^{\prime}f_{h}(e_{1})&=&(e_{1}+\mu_{h}he_{2})\cdot_{h}^{\prime}(e_{1}+\mu_{h}he_{2})\\
    &=&a^{'}_he_{1}+(2\mu_{h}ha^{'}_h+\mu_{h}h^{2}+b_{h}^{'})e_2,\\
    f_{h}(e_{1}\cdot_h e_{1})&=&a_{h}e_{1}+(\mu_{h}ha_{h}+b_{h}\epsilon_{h})e_2.
\end{eqnarray*}
Note that $A_h$ is a finitely generated $\mathbb{C}_h$-module with
a basis $\{(e_{1})_{h}, (e_{2})_{h}\}$.
    Then the ``only if" part of (\ref{e2}) is a consequence of
    $$f_{h}(e_{1})\cdot_{h}^{\prime}f_{h}(e_{1})=f_{h}(e_{1}\cdot_h e_{1}).$$

    Conversely, assume that (\ref{e2-1}) and (\ref{e2-2}) hold. Define a $\mathbb{C}_h$-linear map $f_{h}: A_{h}^{a_{h},b_{h}}\to A_{h}^{a_{h}^{\prime},b_{h}^{\prime}}$ by (\ref{df}).
    Thus $$f_{h}(x)\equiv x\pmod h,\;\forall x\in A.$$
    It is straightforward to check
    \[f_{h}(e_{i})\cdot_{h}^{\prime}f_{h}(e_{j})=f_{h}(e_{i}\cdot_{h}e_{j}),\; 1\leq i,j\leq 2.\]
    Hence $A_{h}^{a_{h},b_{h}}$ and $A_{h}^{a_{h}^{\prime},b_{h}^{\prime}}$ are equivalent.
\end{proof}
\begin{cor}\label{corq2}
    \begin{enumerate}
        \item\label{a} Any quantization of $A^{0,0}$ is equivalent to one of the following mutually non-equivalent cases:
        \begin{enumerate}
            \item $A_{h}^{-h,b_{m}h^{m}}$: $e_{1}\cdot_{h}e_{1}=-he_{1}+b_{m}h^{m}e_{2},\;e_{1}\cdot_{h}e_{2}=0,\;e_{2}\cdot_{h}e_{1}=-he_{2},\;e_{2}\cdot_{h}e_{2}=0$, \\where $m\in\mathbb{Z}^{+},\;0\neq b_m\in\mathbb{C}$;\label{i}
            \item $A_{h}^{a_h,0}$: $e_{1}\cdot_{h}e_{1}=a_{h}e_{1},\;e_{1}\cdot_{h}e_{2}=(a_{h}+h)e_{2},\;e_{2}\cdot_{h}e_{1}=a_{h}e_{2},\;e_{2}\cdot_{h}e_{2}=0$,\\ where $a_{h}\in h\mathbb{C}_{h}$;\label{ii}
            \item $A_{h}^{a_h,b_{1}h}$: $e_{1}\cdot_{h}e_{1}=a_{h}e_{1}+b_{1}he_{2},\;e_{1}\cdot_{h}e_{2}=(a_{h}+h)e_{2},\;e_{2}\cdot_{h}e_{1}=a_{h}e_{2},\;e_{2}\cdot_{h}e_{2}=0$,\\ where $-h\neq a_{h}\in h\mathbb{C}_{h},\;0\neq b_1\in\mathbb{C}$.\label{iii}
        \end{enumerate}
    \item\label{b} Any quantization of $A^{0,1}$ is equivalent to one of the following mutually non-equivalent quantizations of $A^{0,1}$:
    $$A_{h}^{a_h,1}:
    e_{1}\cdot_{h}e_{1}=a_{h}e_{1}+e_{2},\;e_{1}\cdot_{h}e_{2}=(a_{h}+h)e_{2},\;e_{2}\cdot_{h}e_{1}=a_{h}e_{2},\;e_{2}\cdot_{h}e_{2}=0,$$where $a_{h}\in h\mathbb{C}_{h}$.
    \item\label{c} Any quantization of $A^{\lambda,0}$ with $\lambda\ne 0$ is equivalent to one of the following mutually non-equivalent quantizations of $A^{\lambda,0}$:
    $$A_{h}^{\lambda_h,0}:
    e_{1}\cdot_{h} e_{1}=\lambda_{h}e_{1},\; e_{1}\cdot_{h}e_{2}=(\lambda_{h}+h)e_{2},\;e_{2}\cdot_{h}e_{1}=\lambda_{h}e_{2},\;e_{2}\cdot_{h}e_{2}=0,$$ where $\lambda_{h}\in \mathbb{C}_{h}$ satisfying $\lambda_{h}\equiv\lambda\pmod{h}$.
    \end{enumerate}
\end{cor}
\delete{
\begin{ex}
\begin{itemize}
    \item Quantizations of (\ref{c1}).
    \begin{enumerate}
        \item[(1)]\label{(1)} $A_{h}^{-h,bh^{m}}$: $e_{1}\cdot_{h}e_{1}=-he_{1}+bh^{m}e_{2},\;e_{1}\cdot_{h}e_{2}=0,\;e_{2}\cdot_{h}e_{1}=-he_{2},\;e_{2}\cdot_{h}e_{2}=0$, where $ b\in\mathbb{K},\;m\in\mathbb{Z}^{+}$;
        \item[(2)] $A_{h}^{a_h,bh}$: $e_{1}\cdot_{h}e_{1}=a_{h}e_{1}+bhe_{2},\;e_{1}\cdot_{h}e_{2}=(a_{h}+h)e_{2},\;e_{2}\cdot_{h}e_{1}=a_{h}e_{2},\;e_{2}\cdot_{h}e_{2}=0$, where $-h\neq a_{h}\in h\mathbb{K}_{h}, b\in\mathbb{K}$.
    \end{enumerate}
    \item Quantizations of (\ref{c2}). \\
    $A_{h}^{a_h,1}$:
$e_{1}\cdot_{h}e_{1}=a_{h}e_{1}+e_{2},\;e_{1}\cdot_{h}e_{2}=(a_{h}+h)e_{2},\;e_{2}\cdot_{h}e_{1}=a_{h}e_{2},\;e_{2}\cdot_{h}e_{2}=0$, where $a_{h}\in h\mathbb{C}_{h}$.
        \item Quantizations of (\ref{c3}).\\
        $A_{h}^{\lambda_h,0}$:
    $e_{1}\cdot_{h} e_{1}=\lambda_{h}e_{1},\; e_{1}\cdot_{h}e_{2}=(\lambda_{h}+h)e_{2},\;e_{2}\cdot_{h}e_{1}=\lambda_{h}e_{2},\;e_{2}\cdot_{h}e_{2}=0$, where $\lambda_{h}\in \mathbb{C}_{h}$ satisfying $\lambda_{h}\equiv\lambda\pmod{h}$.
\end{itemize}
\end{ex}
}
\begin{proof}
      (\ref{a}). Let $(A_h,\cdot_h)$ be a quantization of $A^{0,0}$.
      By Proposition~\ref{eq-q}~(\ref{e1}),
there are elements $a_{h}$ and $b_{h}$ in $\mathbb{C}_{h}$ such that
      $A_{h}^{a_h,b_h}$ is a quantization which is equivalent
      to $(A_h,\cdot_{h})$. Since  $A_{h}^{a_h,b_h}$ is a quantization of $A^{0,0}$, we have
     $$a_{h}\equiv b_{h}\equiv 0\pmod h.$$
     \item[(i)] $a_{h}=-h$ and $b_h\neq 0$. \\Let $b^{\prime}_{h}$ be an element in $h\mathbb{C}_h$. Then  by Proposition~\ref{eq-q}~(\ref{e2}), $A_{h}^{-h,b_{h}}$ and $A_{h}^{-h,b^{\prime}_{h}}$ are equivalent
     if and only if there is an $\epsilon_{h}\equiv 1\pmod{h}\in \mathbb{C}_{h}$ such that
     \begin{equation}
        b^{\prime}_{h}=b_{h}\epsilon_{h}\label{-h}.
     \end{equation}
   Choose $m\in\mathbb{Z}^{+}$ such that $b_{h}\in
h^{m}\mathbb{C}_{h}\setminus h^{m+1}\mathbb{C}_{h}$. Then let
$b_m$ be a nonzero element in $\mathbb{C}$ satisfying
    $$b_{h}\equiv b_mh^{m}\pmod{h^{m+1}}.$$
    Note that $\frac{b_m h^m}{b_{h}}$ is a well-defined
    element in
    $\mathbb{C}_h$ and
    $$\frac{b_m h^m}{b_{h}}\equiv 1\pmod h.$$
    Thus we take $\epsilon_{h}=\frac{b_mh^{m}}{b_h}$
     in
    (\ref{-h}) and then show that $A_{h}^{-h,b_{h}}$ and $A_{h}^{-h,b_m h^{m}}$ are equivalent.
\item[(ii)-(1)] $a_{h}=-h$ and $b_h=0$. \\Then $A_{h}^{a_h,b_{h}}=A_{h}^{-h,0}$.
\item[(ii)-(2)] $a_{h}\neq-h$ and $b_{h}\in h^{2}\mathbb{C}_{h}$. \\
Then $\frac{b_{h}}{h(a_{h}+h)}$ is a well-defined element in $\mathbb{C}_{h}$. Thus $A_{h}^{a_h,b_{h}}$ is equivalent to $A_{h}^{a_h,0}$ by taking $\epsilon_{h}=1$ and $\mu_{h}=\frac{b_{h}}{h(a_{h}+h)}$ in (\ref{eqv}).
\item[(iii)] $a_{h}\neq-h$ and $b_{h}\in h\mathbb{C}_{h}\setminus h^{2}\mathbb{C}_{h}$.\\ Let $b_1$ be the nonzero element in $\mathbb{C}$ that satisfies
$$b_{h}\equiv b_1 h\pmod{h^{2}}.$$ Thus $\frac{b_1 h}{b_{h}}$ is a well-defined element in $\mathbb{C}_h$ and
$$\frac{b_1 h}{b_{h}}\equiv 1\pmod h.$$
So $A_{h}^{a_h,b_{h}}$ and $A_{h}^{a_h,b_1 h}$ are equivalent by taking $\epsilon_{h}=\frac{b_1 h}{b_{h}}$ and $\mu_{h}=0$ in
(\ref{eqv}).\\
Note that these cases are mutually non-equivalent. Thus the conclusion follows.

(\ref{b}) and (\ref{c}) follow by similar arguments.
\end{proof}

 Since all 2-dimensional transposed Poisson algebras with
    non-abelian Lie algebras are of Novikov-Poisson type, we list the
    corresponding quantizations given by (\ref{eq-dnp}) as
    follows.
    \begin{pro}
       \begin{enumerate}
        \item For $A^{0,0}$ which is the sub-adjacent transposed Poisson algebra of the Novikov-Poisson algebra given by
        \begin{eqnarray*}
            &&e_i\cdot e_j=0,\;1\le i,j\le 2,\\
            &&e_{1}\circ e_{1}=a_1e_{1}+b_1e_{2},\; e_{1}\circ e_{2}=(a_1+1)e_{2},\;e_{2}\circ e_{1}=a_1e_{2},\;e_{2}\circ e_{2}=0,
        \end{eqnarray*}
        where $a_1,b_1\in\mathbb{C}$, the corresponding quantization given by {\rm (\ref{eq-dnp})} is exactly $A_{h}^{a_1h,b_1h}$.
        \item For $A^{0,1}$ which is the sub-adjacent transposed Poisson algebra of the Novikov-Poisson algebra given by
        \begin{eqnarray*}
        &&e_1\cdot e_1=e_2,\;e_1\cdot e_2=e_2\cdot e_1=e_2\cdot e_2=0,\\
    &&e_{1}\circ e_{1}=a_1e_{1}+b_1e_{2},\; e_{1}\circ e_{2}=(a_1+1)e_{2},\;e_{2}\circ e_{1}=a_1e_{2},\;e_{2}\circ e_{2}=0,
        \end{eqnarray*}
        where $a_1,b_1\in\mathbb{C}$, the corresponding quantization given by {\rm (\ref{eq-dnp})} is exactly $A_{h}^{a_1h,1+b_1h}$ which is equivalent to $A_{h}^{a_1h,1}$.
        \item For $A^{\lambda,0}$ with $\lambda\ne 0$ which is the sub-adjacent transposed Poisson algebra of the Novikov-Poisson algebra given by
        \begin{eqnarray*}
        &&e_{1}\cdot e_{1}=\lambda e_{1},\; e_{1}\cdot e_{2}=e_{2}\cdot e_{1}=\lambda e_{2},\;e_2\cdot e_2=0,\\
        &&e_{1}\circ e_{1}=\lambda_{1}e_{1}+b_1e_{2},\; e_{1}\circ e_{2}=(\lambda_{1}+1)e_{2},\;e_{2}\circ e_{1}=\lambda_{1}e_{2},\;e_{2}\circ e_{2}=0,
        \end{eqnarray*}
        where $\lambda_{1},b_1\in\mathbb{C}$, the corresponding quantization given by {\rm (\ref{eq-dnp})}
         is exactly $A_{h}^{\lambda+\lambda_{1}h,b_1h}$ which is equivalent to $A_{h}^{\lambda+\lambda_{1}h,0}$.
    \end{enumerate}
    \end{pro}
    \delete{
\begin{rmk}
    Since all 2-dimensional transposed Poisson algebras with
    non-abelian Lie algebras are of Novikov-Poisson type, we list the
    corresponding quantizations given by (\ref{eq-dnp})
    as
    follows.
     \begin{itemize}
        \item For $A^{0,0}$ which is the sub-adjacent transposed Poisson algebra of the Novikov-Poisson algebra given by
        \begin{eqnarray*}
            &&e_i\cdot e_j=0,\;1\le i,j\le 2,\\
            &&e_{1}\circ e_{1}=a_1e_{1}+b_1e_{2},\; e_{1}\circ e_{2}=(a_1+1)e_{2},\;e_{2}\circ e_{1}=a_1e_{2},\;e_{2}\circ e_{2}=0,
        \end{eqnarray*}
        where $a_1,b_1\in\mathbb{K}$, the corresponding quantization given by (\ref{eq-dnp}) is exactly $A_{h}^{a_1h,b_1h}$.
        \item For $A^{0,1}$ which is the sub-adjacent transposed Poisson algebra of the Novikov-Poisson algebra given by
        \begin{eqnarray*}
        &&e_1\cdot e_1=e_2,\;e_1\cdot e_2=e_2\cdot e_1=e_2\cdot e_2=0,\\
    &&e_{1}\circ e_{1}=a_1e_{1}+b_1e_{2},\; e_{1}\circ e_{2}=(a_1+1)e_{2},\;e_{2}\circ e_{1}=a_1e_{2},\;e_{2}\circ e_{2}=0,
        \end{eqnarray*}
        where $a_1,b_1\in\mathbb{K}$, the corresponding quantization given by (\ref{eq-dnp}) is exactly $A_{h}^{a_1h,1+b_1h}$ which is equivalent to $A_{h}^{a_1h,1}$.
        \item For $A^{\lambda,0}$ with $\lambda\ne 0$ which is the sub-adjacent transposed Poisson algebra of the Novikov-Poisson algebra given by
        \begin{eqnarray*}
        &&e_{1}\cdot e_{1}=\lambda e_{1},\; e_{1}\cdot e_{2}=e_{2}\cdot e_{1}=\lambda e_{2},\;e_2\cdot e_2=0,\\
        &&e_{1}\circ e_{1}=\lambda_{1}e_{1}+b_1e_{2},\; e_{1}\circ e_{2}=(\lambda_{1}+1)e_{2},\;e_{2}\circ e_{1}=\lambda_{1}e_{2},\;e_{2}\circ e_{2}=0,
        \end{eqnarray*}
        where $\lambda_{1},b_1\in\mathbb{K}$, the corresponding quantization given by (\ref{eq-dnp}) is exactly $A_{h}^{\lambda+\lambda_{1}h,b_1h}$ which is equivalent to $A_{h}^{\lambda+\lambda_{1}h,0}$.
    \end{itemize}
\end{rmk}
} \noindent{\bf Acknowledgments.} This work is supported by NSFC
(11931009, 12271265, 12261131498, 12326319), the Fundamental
Research Funds for the Central Universities and Nankai Zhide
Foundation. The authors thank the referees for valuable
suggestions.

\smallskip

\noindent {\bf Declaration of interests. } The authors have no
conflicts of interest to disclose.

\smallskip

\noindent {\bf Data availability. } Data sharing is not applicable
to this article as no new data were created or analyzed in this
study.

\vspace{-.5cm}


\end{document}